\documentclass[]{article}
\usepackage{proceed2e}
\usepackage{times}

\newcount\Comments  
\usepackage{times}

\usepackage[font=small]{caption}
\usepackage{latexsym, epic, eepic,graphics,graphicx,mathtools}
\usepackage{amsmath,amssymb,amsthm}
\usepackage{bm}
\usepackage[square,comma,numbers,sort&compress]{natbib}
\usepackage{algorithmic}
\usepackage{enumerate}
\usepackage[usenames,dvipsnames,svgnames,table]{xcolor}
\usepackage{hyperref}
\hypersetup{colorlinks,
            linkcolor=blue,
            citecolor=blue,
            urlcolor=magenta,
            linktocpage,
            plainpages=false}
\usepackage{xspace}
\usepackage{booktabs}

\usepackage{tikz,pgfplots}
\usetikzlibrary{arrows,shapes,positioning}
\usetikzlibrary{shapes.multipart}

\newcommand{\figsqueeze}{}
\newcommand{\ignore}[1]{}
\definecolor{darkgreen}{rgb}{0,0.5,0}
\definecolor{darkred}{rgb}{0.7,0,0}
\newcommand{\kibitz}[2]{\ifnum\Comments=1\textcolor{#1}{#2}\fi}
\newcommand{\jenn}[1]{\kibitz{magenta}{[JWV: #1]}}
\newcommand{\raf}[1]{\kibitz{darkgreen}{[RMF: #1]}}
\newcommand{\miro}[1]{\kibitz{Cerulean}{[MD: #1]}}

\newcommand{\argmax}{\mathop{\rm argmax}}
\newcommand{\argmin}{\mathop{\rm argmin}}

\newcommand{\Parens}[1]{\left(#1\right)}
\newcommand{\BigParens}[1]{\Bigl(#1\Bigr)}
\newcommand{\bigParens}[1]{\bigl(#1\bigr)}
\newcommand{\Bracks}[1]{\left[#1\right]}
\newcommand{\BigBracks}[1]{\Bigl[#1\Bigr]}

\newcommand{\bigBracks}[1]{\bigl[#1\bigr]}
\newcommand{\set}[1]{\{#1\}}
\newcommand{\Set}[1]{\left\{#1\right\}}

\newcommand{\reals}{\mathbb{R}}
\newcommand{\realsplusinf}{(-\infty,\infty]}
\newcommand{\ones}{\mathbb{I}}

\def\mythin{1.5mu minus 1mu}

\newcommand{\thinskips}{%
\thinmuskip=\mythin\relax
\medmuskip=\mythin\relax
\thickmuskip=\mythin\relax
}

\theoremstyle{plain}
\newtheorem{theorem}{Theorem}
\newtheorem{proposition}{Proposition}
\newtheorem{definition}{Definition}
\newtheorem{lemma}{Lemma}
\newtheorem{corollary}{Corollary}

\theoremstyle{definition}
\newtheorem{example}{Example}

\newcommand{\squishlist}{
   \begin{list}{$\bullet$}
    { \setlength{\itemsep}{0pt}      \setlength{\parsep}{1pt}
      \setlength{\topsep}{0pt}       \setlength{\partopsep}{0pt}
     \setlength{\leftmargin}{1em} \setlength{\labelwidth}{1.5em}
      \setlength{\labelsep}{0.5em} } }
\newcommand{\squishend}{  \end{list}  }

\newcommand{\romanbf}[1]{{\boldsymbol{#1}}}

\def\x{\romanbf{x}}
\def\p{\romanbf{p}}
\def\q{\romanbf{q}}
\def\r{\romanbf{r}}
\def\s{\romanbf{s}}
\def\1{\mathbf{1}}
\def\0{\mathbf{0}}

\def\a{\romanbf{a}}

\def\A{\mathbf{A}}
\def\E{\mathbb{E}}

\def\I{\mathcal{I}}

\def\M{\mathcal{M}}

\def\reals{\mathbb{R}}

\def\y{\romanbf{y}}

\def\rhob{\boldsymbol{\rho}}

\newcommand{\G}{\mathcal{G}}
\newcommand{\X}{\mathcal{X}}
\newcommand{\Y}{\mathcal{Y}}
\let\sectsymb\S
\renewcommand{\S}{\mathcal{S}}
\renewcommand{\b}{\romanbf{b}}
\newcommand{\KL}{\text{\upshape KL}}

\def\payb{\rhob}
\def\lub{\boldsymbol{\lambda}}
\def\mub{\boldsymbol{\mu}}
\def\hmu{\hat{\mu}}
\def\hmub{\boldsymbol{\hat{\mu}}}
\def\tmub{\boldsymbol{\tilde{\mu}}}
\def\hull{\mathcal{M}}

\newcommand{\nub}{\boldsymbol{\nu}}

\newcommand{\C}{\mathbf{C}}
\newcommand{\tC}{\tilde{C}}
\newcommand{\tR}{\tilde{R}}
\newcommand{\tD}{\tilde{D}}

\newcommand{\tqb}{\tilde{\romanbf{q}}}
\newcommand{\tpb}{\tilde{\romanbf{p}}}
\newcommand{\tsb}{\tilde{\romanbf{s}}}

\newcommand{\hb}{\hat{b}}
\newcommand{\hbb}{\romanbf{\hat{b}}}

\newcommand{\etab}{\boldsymbol{\eta}}
\newcommand{\deltab}{\boldsymbol{\delta}}
\newcommand{\omegab}{\boldsymbol{\omega}}
\newcommand{\Rx}{R_{\oplus}}
\newcommand{\Cx}{C_{\oplus}}
\newcommand{\Dx}{D_{\oplus}}

\newcommand{\tCx}{\tC_{\oplus}}
\newcommand{\tDx}{\tD_{\oplus}}
\newcommand{\CCx}{\mathbf{C}_{\oplus}}

\newcommand{\hullx}{\hull_{\oplus}}
\newcommand{\trb}{\tilde{\romanbf{r}}}
\newcommand{\tN}{{\tilde{N}}}
\newcommand{\ini}{\mathtt{ini}}
\newcommand{\fin}{\mathtt{fin}}
\newcommand{\event}{\mathcal{E}}

\DeclareMathOperator{\relint}{relint}
\DeclareMathOperator{\dom}{dom}

\DeclareMathOperator{\conv}{conv}
\DeclareMathOperator{\epi}{epi}
\DeclareMathOperator{\cl}{cl}

\newcommand{\inprod}{\cdot}
\newcommand{\wo}{\backslash}
\newcommand{\given}{\mathrel{\vert}}
\newcommand{\Sec}[1]{Sec.~\ref{sec:#1}}
\newcommand{\SEC}[1]{SECTION~\ref{sec:#1}}
\newcommand{\App}[1]{Appendix~\ref{app:#1}}
\newcommand{\Eq}[1]{Eq.~\eqref{eq:#1}}
\newcommand{\Prop}[1]{Prop.~\ref{prop:#1}}
\newcommand{\PROP}[1]{PROPOSITION~\ref{prop:#1}}
\newcommand{\Thm}[1]{Theorem~\ref{thm:#1}}
\newcommand{\THM}[1]{THEOREM~\ref{thm:#1}}
\newcommand{\Tab}[1]{Table~\ref{tab:#1}}
\newcommand{\Prot}[1]{Protocol~\ref{prot:#1}}
\newcommand{\Lem}[1]{Lemma~\ref{lem:#1}}
\newcommand{\LEM}[1]{LEMMA~\ref{lem:#1}}
\newcommand{\Fig}[1]{Fig.~\ref{fig:#1}}
\newcommand{\Ex}[1]{Example~\ref{ex:#1}}
\newcommand{\Def}[1]{Definition~\ref{def:#1}}

\newcommand{\defline}[1]{\smallskip\par\centerline{$#1$}}
\newcommand{\Slovakia}{Norway\xspace}
\newcommand{\CondPrice}{\textsc{CondPrice}\xspace}

\newcommand{\CondVal}{\textsc{Ex\-Util}\xspace}
\newcommand{\DecVal}{\textsc{DecUtil}\xspace}
\newcommand{\ZeroVal}{\textsc{Ze\-ro\-Util}\xspace}
\newcommand{\Price}{\textsc{Price}\xspace}
\newcommand{\WCL}{\ensuremath{\mathtt{WCLoss}}}

\newcommand{\AdvanceState}{\ensuremath{\mathtt{NewState}}}
\newcommand{\AdvanceCost}{\ensuremath{\mathtt{NewCost}}}
\newcommand{\Value}{\ensuremath{\mathtt{Util}}}
\newcommand{\tValue}{\ensuremath{\mathtt{\tilde{U}til}}}
\newcommand{\RBV}{\ensuremath{\mathtt{RBUtil}}}
\newcommand{\tRBV}{\ensuremath{\mathtt{\tilde{RB}Util}}}
\newcommand{\ttt}{\tilde{t}}

\renewcommand{\u}{\romanbf{u}}
\renewcommand{\v}{\romanbf{v}}
\newcommand{\talpha}{\alpha}

\newcommand{\ttC}{\tC}
\newcommand{\ttR}{\tR}
\newcommand{\ttD}{\tD}
\newcommand{\ttCx}{\tCx}
\newcommand{\ttDx}{\tDx}
\newcommand{\ttValue}{\tValue}

\newenvironment{protocol}{\begingroup%
\begin{algorithm}}{\end{algorithm}\endgroup}

\renewcommand{\=}{\!=\!}

\newcount\Spacehack  
\ifnum\Spacehack=1
\let\oldsection\section
\let\oldsubsection\subsection
\renewcommand{\section}[1]{\vspace{-5pt}\oldsection{#1}\vspace{-10pt}}
\renewcommand{\subsection}[1]{\vspace{-5pt}\oldsubsection{#1}\vspace{-5pt}}
\fi

\usepackage[ruled]{algorithm2e}

\SetAlFnt{\small}
\SetAlCapFnt{\small}
\SetAlCapNameFnt{\small}
\SetAlCapHSkip{0pt}
\IncMargin{-\parindent}

\title{Market Making with Decreasing Utility for Information}

\author{Miroslav Dud\'ik
\\{Microsoft Research}
\And
Rafael Frongillo
\\{Microsoft Research}
\And
Jennifer Wortman Vaughan
\\{Microsoft Research}}

\begin{document}

\maketitle

\begin{abstract}
\jenn{Added back this weird "func\-tion" to
make the arxiv version match the camera ready.  We should remove
this for the long version.}
  We study information elicitation in cost-func\-tion-based
  combinatorial prediction markets when the market maker's utility for information
  decreases over time.  In the \emph{sudden
    revelation} setting, it is known that some piece of information
  will be revealed to traders, and the market maker wishes to prevent
  guaranteed profits for trading on the sure information.
  In the \emph{gradual decrease} setting,
  the market maker's utility for (partial) information decreases
  continuously over time. We design adaptive cost functions for both
  settings which: (1) preserve the information previously gathered in
  the market; (2) eliminate (or diminish) rewards to traders for the publicly
  revealed information; (3) leave the
  reward structure unaffected for other information; and (4)
  maintain the market maker's worst-case loss. Our constructions
  utilize mixed Bregman divergence, which matches our notion of
  utility for information.
\end{abstract}

\section{INTRODUCTION}
\label{sec:intro}

Prediction markets have been used to elicit information in a variety of domains, including business~\cite{Spa:03,Cowgill:08,C07,B13}, politics~\cite{Berg+08,U+12}, and entertainment~\cite{PenScience:01}.
In a prediction market, traders buy and sell \emph{securities} with values that depend on some unknown future outcome.  For example, a market might offer securities worth \$1 if \Slovakia wins a gold medal in Men's Moguls in the 2014 Winter Olympics and \$0 otherwise.
Traders are given an incentive to reveal their beliefs about the outcome by buying and selling securities, e.g., if the current price of the above security is \$0.15, traders who believe that the probability of \Slovakia winning 
is more than 15\% are incentivized to buy and those who believe that the probability is less than 15\% are incentivized to sell. The equilibrium price 
reflects the market consensus about the security's expected payout (which here coincides with the probability of \Slovakia winning the medal).

There has recently been a surge of research on the design of prediction markets operated by a centralized authority called a \emph{market maker}, an algorithmic agent that offers to buy or sell securities at some current 
price that depends on the history of trades in the market.
Traders in these markets can express their belief
whenever it differs from the current price
by either buying or selling, regardless of whether other traders are willing to act as a counterparty,
because the market maker always acts as a counterparty, thus ``providing the liquidity''
and subsidizing the information collection.
This 
is useful in situations when the lack of interested
traders would negatively impact the efficiency
in a traditional exchange.
%
 Of particular interest to us are \emph{combinatorial prediction markets}~\citep{H03,H07,CFNP07,CFLPW08,CGP08,GP09,PX11} which offer securities on various related events such as ``\Slovakia wins a total of 4 gold medals in the 2014 Winter Olympics'' and ``\Slovakia wins a gold medal in Men's Moguls.''  In combinatorial markets with large, expressive security spaces, such as an Olympics market with securities covering 88 nations participating in 98 events, the lack of an interested counterparty is a major concern. Only a single trader may be interested in trading the security associated with a specific event, but we would still like the market to incorporate this trader's information.
\miro{I am hoping that the reasoning in the previous two sentences is clearer given the previous prose. Do we need to be more precise about ``liquidity'' by saying something like ``interested counterparty''? Is it clear how the incorporation of trader's information could be hampered?}

Most market makers considered in the literature are implemented using a pricing function called the \emph{cost function}~\cite{Chen07}.
While such markets have many favorable properties~\cite{Abernethy11,Abernethy13},
the current approaches have several drawbacks that limit their applicability in real-world settings.
First, existing work implicitly assumes that the outcome is revealed all at once. When concerned about ``just-in-time arbitrage,''
in which traders closer to the information source
make last-minute guaranteed profits by trading on the sure information before the market maker can adjust prices,
the market maker can prevent such profits by closing the entire market just before the outcome is revealed.
%
%
This approach is undesirable when partial information about the outcome is revealed over time, as is often the case
in practice, including the Olympics market.
For instance, we may learn the results of Men's Moguls before Ladies' Figure Skating has taken place.  Closing a large combinatorial market whenever a small portion of the outcome is determined seems to be an unreasonably large intervention.


Second, in real markets, the information captured by the market's consensus prices often becomes less useful as the revelation of the outcome approaches.  Consider a market over the event ``Unemployment in the U.S.\ falls below 5.8\% by the end of 2015.''  Although there may be a particular moment when the unemployment rate is publicly revealed, this information becomes gradually less useful as that moment approaches; the government may be less able to act on the information as the end of the year draws near.
 In the Olympics market, the outcome of a particular competition is often more certain as the final announcement approaches, e.g., if one team is far ahead by the half-time of a hockey game, market forecasts become less interesting.
Existing market makers fail to take this diminishing utility for information into account, with the strength of the market incentives remaining constant over time.


To address these two shortcomings of existing markets, we consider two settings:
\squishlist
\item a \emph{sudden revelation} setting in which it is known that some piece of information (such as the winner of Men's Moguls) will be publicly revealed at a particular time,
    driving the market maker's utility for this information to zero; crucially, in this setting we assume that the market maker \emph{does not} have direct access to this information at the time it is revealed, which is realistic in the case of the Olympics where a human might not be available to input winners for all 98 events in real time;
\item a \emph{gradual decrease} setting in which the market maker has a diminishing utility for a piece of information (such as the unemployment rate for 2015) over time and therefore is increasingly unwilling to pay for this information even while other information remains valuable.
\squishend
The sudden revelation setting can be viewed as
a special case of the gradual decrease setting.
In both cases, we model the relevant information as a 
variable $X$, representing a partly determined outcome
such as the identity of the gold medal winner in a single
sports event.

We consider cost-function-based market makers in which the cost function switches one or many times,
and aim to design switching strategies such that:
(1) information previously gathered in the market is not lost at the time of the switch,
(2) a trader who knows the value of~$X$ but has no additional information
is unable to profit after the switch (for the sudden revelation setting) or
is able to profit less and less over
time (in the gradual decrease setting), and
(3) the market maker maintains the same reward structure for any other information that
traders may have. 
\jenn{In the long version we could consider adding bounded loss as an explicit desideratum since we'll have space to discuss it more.}
To formalize these objectives, we define the notion of the market maker's
utility (\Sec{formalism}) and show how it corresponds to
the \emph{mixed Bregman divergence}~\citep{DudikLaPe12,Gordon99} (\Sec{bregman}).

For the sudden revelation setting (\Sec{one-step}), we introduce a
generic cost function switching technique
which in many cases removes the rewards for ``just-in-time arbitragers'' who know only the value of $X$, while allowing traders with other information to profit, satisfying our objectives.

For the gradual decrease setting (\Sec{gradual}),
we focus on 
\emph{linearly constrained market makers} (LCMMs) \citep{DudikLaPe12}, proposing a time-sensitive market maker that
gradually decreases liquidity by employing the cost function of a different LCMM at each point in time, again
meeting our objectives.

Others have considered the design of cost-function-based markets with adaptive liquidity~\cite{LV13,OS11,OSPR10,OS12,A+14}.  That line of research has typically focused on the goal of slowing down price movement as more money enters the market.  In contrast, we adjust liquidity to reflect the current market maker's utility which can be viewed as something external to trading in the market.  Additionally, we change liquidity only in the ``low-utility'' parts of the market, whereas previous work considered market-wide liquidity shifts.  \citet{B+12} designed a Bayesian market maker that adapts to perceived increases in available information.  Our market maker does not try to infer high information periods, but assumes that a schedule of public revelations is given a priori.  Our market makers have guaranteed bounds on worst-case loss whereas those of \citet{B+12} do not.

\section{SETTING AND DESIDERATA}
\label{sec:formalism}

We begin by reviewing cost-function-based market making before
describing our desiderata.  Here and throughout the paper we make
use of many standard results from convex analysis,
summarized in \App{convex}. All of the proofs in this paper are
relegated to the appendix.\ \footnote{%
The full version of this paper
on arXiv includes the appendix.
}

\subsection{COST-FUNCTION-BASED MARKET MAKING}
\label{sec:formalism-cost-function-based}

Let $\Omega$ denote the \emph{outcome space}, a finite set of mutually
exclusive and exhaustive states of the world.  We are interested in
the design of cost-function-based market makers operating over a set
of $K$ \emph{securities} on $\Omega$ specified by a \emph{payoff
  function} $\payb: \Omega \rightarrow \reals^K$, where
$\payb(\omega)$ denotes the vector of security payoffs if the outcome
$\omega \in \Omega$ occurs.  Traders may purchase \emph{bundles} $\r
\in \reals^K$ of securities from the market maker, with $r_i$ denoting
the quantity of security $i$ that the trader would like to purchase;
negative values of $r_i$ are permitted and represent short selling.  A
trader who purchases a bundle $\r$ of securities pays a specified cost
for this bundle up front and receives a (possibly negative) payoff of
$\payb(\omega) \cdot \r$ if the outcome $\omega \in \Omega$ occurs.

Following \citet{Chen07} and \citet{Abernethy11,Abernethy13}, we
assume that the market maker initially prices securities using a
convex potential function $C:\reals^K\to\reals$, called the
\emph{cost function}.  The current state of the market is summarized
by a vector $\q \in \reals^K$, where $q_i$ denotes the total number of
shares of security $i$ that have been bought or sold so far.  If the
market state is $\q$ and a trader purchases the bundle $\r$, he must pay the market maker $C(\q + \r) -
C(\q)$.  The new market state is then $\q + \r$.  The
\emph{instantaneous price} of security $i$ is $\partial C(\q)
/ \partial q_i$ whenever well-defined; this
is the price per share of an infinitesimally small quantity
of security $i$, and is frequently interpreted as the traders'
collective belief about the expected payoff of this security.
Any expected payoff must lie
in the convex hull of the set $\{\payb(\omega)\}_{\omega \in \Omega}$,
called \emph{price space},
denoted $\hull$.

While our cost function might not be differentiable at all states $\q$,
it is always \emph{subdifferentiable} thanks to convexity,
i.e., its subdifferential $\partial C(\q)$
is non-empty for each $\q$ and, if it is a singleton, it
coincides with the gradient. Let $\p(\q)\coloneqq\partial C(\q)$
be called the \emph{price map}. The set $\p(\q)$ is always convex and can
be viewed as a multi-dimensional version of the ``bid-ask spread''.
In a state $\q$, a trader can make an expected profit
if and only if he believes that $\E[\rhob(\omega)]\not\in\p(\q)$.
If $C$ is differentiable at $\q$, we slightly abuse notation and
also use $\p(\q)\coloneqq\nabla C(\q)$.

We assume that the cost function satisfies two standard properties:
\emph{no arbitrage} and \emph{bounded loss}.
The former means that as long as all outcomes $\omega$ are possible, there
are no market transactions with a guaranteed profit for a trader.
The latter means that the worst-case loss of the market maker is a priori
bounded by a constant. Together, they imply that
the cost function $C$ can be written in the form
$C(\q) = \sup_{\mub \in \hull} [\mub \cdot \q - R(\mub)]$,
where $R$ is the convex conjugate of $C$, with $\dom R=\hull$.
See
\citet{Abernethy11,Abernethy13} for an analysis of the
properties of such markets.

\begin{example}\emph{Logarithmic market-scoring rule (LMSR).}
\label{ex:formalism-lmsr}
The LMSR of \citet{H03,H07}
is a cost function for a \emph{complete market} where traders can
express any probability distribution over $\Omega$. Here, for any $K\ge 1$,
$\Omega=[K]\coloneqq\{1,\dotsc,K\}$
and $\rho_i(\omega)=\1[i=\omega]$ where $\1[\cdot]$ is a 0/1 indicator, i.e.,
the security $i$ pays out \$1 if the outcome $i$ occurs and \$0 otherwise.
The price space $\hull$ is the simplex of probability distributions
in $K$ dimensions. The cost function is
$C(\q)=\ln\bigParens{\sum_{i=1}^K e^{q_i}}$,
which is differentiable and generates
prices
$p_i(\q)=e^{q_i} / \bigParens{\sum_{j=1}^K e^{q_j}}$. Here $R$ is the negative
entropy function, $R(\mub) = \sum_{i=1}^K \mu_i \ln \mu_i$.
\end{example}

\begin{example}\emph{Square.}
\label{ex:formalism-square}
The square market consists of two independent securities ($K=2$)
each paying out either \$0 or \$1. This can be encoded as
$\Omega=\set{0,1}^2$ with $\rho_i(\omegab)=\omega_i$ for $i=1,2$.
The price space is the unit square $\hull=[0,1]^2$. Consider the cost function
$C(\q)=\ln\bigParens{1+e^{q_1}}+\ln\bigParens{1+e^{q_2}}$,
which is differentiable and generates prices $p_i(\q)=e^{q_i}/(1+e^{q_i})$
for $i=1,2$. Using this cost function is equivalent to running two independent binary markets,
each with an LMSR cost function. We have $R(\mub) = \sum_{i=1}^2 \mu_i \ln \mu_i + (1-\mu_i)
\ln (1-\mu_i)$.
\end{example}

\begin{example}\emph{Piecewise linear cost.}
\label{ex:formalism-piecewise-linear}
Here we describe a non-differentiable cost function for a single
binary security ($K=1$). Let
$\Omega=\set{0,1}$ and $\rho(\omega)=\omega$, so $\hull=[0,1]$.
The cost function is
$C(q)=\max\set{0,q}$. It
gives rise to the price map such that $p(q)=0$ if $q<0$,
and $p(q)=1$ if $q>0$, but at $q=0$, we have $p(q)=[0,1]$,
i.e., because of non-differentiability we have a bid-ask spread
at $q=0$.  Here, $R(\mu)=\ones\bigBracks{\mu\in[0,1]}$ where $\ones[\cdot]$
is a $0/\infty$ indicator, equal to $0$ if true and $\infty$ if
false. This market is uninteresting on its own, but will be useful to
us in \Sec{one-step-examples}.
\end{example}

\begin{protocol}[t]
\caption{Sudden Revelation Market Makers}
\label{prot:switch}
\begin{tabbing}
\textbf{Input:\quad}%
  initial cost function $C$, initial state $\s^\ini$, switch time $t$,\\
\hphantom{\textbf{Input:\quad}}%
  update functions $\AdvanceCost(\q)$, $\AdvanceState(\q)$\\[4pt]
  Until time $t$:\\
\hphantom{---}\=
       sell bundles $\r^1,\dotsc,\r^N$ priced using $C$\\
\>\hphantom{---}\=
          for the total cost $C(\s^\ini\!\!+\!\r)-C(\s^\ini)$
          where $\r=\sum_{i=1}^N \!\r^i$\\
\>
       let $\s=\s^\ini\!+\r$\\
  At time $t$:\\
\>
       $\tC\gets\AdvanceCost(\s)$\\
\>
       $\tsb\gets\AdvanceState(\s)$\\
  After time $t$:\\
\>
       sell bundles $\trb^1,\dotsc,\trb^\tN$ priced using $\tC$\\
\>\>    for the total cost $\tC(\tsb+\trb)-\tC(\tsb)$
          where $\trb=\sum_{i=1}^{\tN} \trb^i$\\
\>
       let $\tsb^\fin=\tsb+\trb$\\
  Observe $\omega$\\
  Pay $(\r+\trb)\inprod\rhob(\omega)$ to traders\\[-15pt]
\end{tabbing}
  \raf{This protocol might confuse people if we don't also include the
    standard cost-function-based protocol (because for simplicity here
    we have to gloss over the telescoping cost and who gets paid
    what)}\jenn{I agree. Expand in long version?}
\end{protocol}

\begin{protocol}[t]
\caption{Gradual Decrease Market Makers}
\label{prot:time-sensitive}
\begin{tabbing}
\textbf{Input:\quad}%
  time-sensitive cost function $\C(\q;\,t)$,\\
\hphantom{\textbf{Input:\quad}}%
  initial state $\s^0$, initial time $t^0$,\\
\hphantom{\textbf{Input:\quad}}%
  update function $\AdvanceState(\q;\,t,t')$\\[4pt]
  For $i=1,\dotsc,N$ (where $N$ is an unknown number of trades):\\
\hphantom{---}\=
      at time $t^i\ge t^{i-1}$: receive a request for a bundle $\r^i$\\
\>
      $\tsb^{i-1}\gets\AdvanceState(\s^{i-1};\,t^{i-1}, t^{i})$\\
\>
      sell the bundle $\r^i$\\
\>\hphantom{---}\=
        for the cost $\C(\tsb^{i-1}+\r^i;\,t^i)-\C(\tsb^{i-1};\,t^i)$\\
\>
      $\s^i\gets\tsb^{i-1}+\r^i$\\
  Observe $\omega$\\
  Pay $\sum_{i=1}^N \r^i\inprod\rhob(\omega)$ to traders\\[-15pt]
\end{tabbing}
\end{protocol}

\subsection{OBSERVATIONS AND ADAPTIVE COSTS}

We study two settings. In the \emph{sudden revelation setting}, it is
known to both the market maker and the traders that at a particular
point in time (the observation time) some information about the
outcome (an observation) will be publicly revealed to the traders, but
not to the market maker.  More precisely, let any function on $\Omega$
be called a \emph{random variable} and its value called the
\emph{realization} of this random variable.  Given a random variable
$X:\Omega\to\X$, we assume that its realization is revealed to the
traders at the observation time.  For a random variable $X$ and a
possible realization $x$, we define the \emph{conditional outcome
  space} by $\Omega^x\coloneqq\set{\omega\in\Omega:\:X(\omega)=x}$.
After observing $X=x$ (where, using standard random variable shorthand, we
write $X$ for $X(\omega)$), the traders can conclude that
$\omega\in\Omega^x$.  Note that the sets
$\set{\Omega^x}_{x\in\X}$ form a partition of $\Omega$.

We design \emph{sudden revelation market makers} (\Prot{switch})
that replace
the cost function $C$ with a new cost function $\tC$, and
the current market state $\s$ (i.e., the current value of
$\q$ in the definition above) with a new market state $\tsb$
in order to reflect the decrease in the utility for information about
$X$. Such a switch would typically occur just before the observation time.
Note that we allow the new cost
function $\tC$ as well as the new state $\tsb$
to be chosen adaptively according to the last state $\s$ of the
original cost function $C$.

In the \emph{gradual decrease} setting, the
utility for information about a future observation $X$ is decreasing
continuously over time.
We use a \emph{gradual decrease market maker}
(\Prot{time-sensitive}) with a time-sensitive cost function $\C(\q;\,t)$ which
sells a bundle $\r$ for the cost $\C(\q+\r;\,t)-\C(\q;\,t)$ at
time~$t$, when the market is in a state $\q$.
We place no assumptions on $\C$ other than that for each $t$,
the function $\C(\cdot;\,t)$ should be
an arbitrage-free bounded-loss cost function.
The market maker may modify the state between the trades.

\Prot{time-sensitive}
alternates between trades and cost-function switches akin to those
in \Prot{switch}.
In each iteration $i$, the cost function
$\C(\cdot;\,t^{i-1})$ is replaced by the cost function $\C(\cdot;\,t^i)$ while
simultaneously replacing the state $\s^{i-1}$ by the state $\tsb^{i-1}$.
Crucially, unlike \Prot{switch}, the
cost-function switch here is \emph{state independent},
so any state-dependent adaptation happens through
the state update.~\footnote{%
This simplifying restriction
matches our
solution concept in \Sec{gradual}, but it
could be dropped for greater generality.}%
\jenn{Reading
  this makes me say ``why?'' Is there an easy reason we can give?}
  \miro{Does this explanation do the trick?}

At a high level, within each of the protocols, our goal is to
design switch strategies that satisfy the following criteria:

\squishlist

\item Any information that has already been gathered from traders
  about the relative likelihood of the outcomes in the conditional
  outcome spaces is preserved.

\item
  A trader who has information
  about the observation $X$ but has no additional information about
  the relative likelihood of outcomes in the conditional outcome
  spaces is unable to profit
  from this information
  (for sudden revelation), or the profits
  of such a trader are
  decreasing over time
  (for gradual decrease).

\item The market maker continues to reward traders for new information
  about the relative likelihood of outcomes in the conditional outcome
  spaces as it did before, with prices reflecting the market maker's utility
  for information within these sets of outcomes.

\squishend

To reason about these goals, it is necessary to define what
we mean by the information that has been gathered in the market and
the market maker's utility.

\subsection{MARKET MAKER'S UTILITY}
\label{sec:val:info}

By choosing a cost function, the market maker creates an incentive
structure for the traders. Ideally, this incentive structure should be
aligned with the market maker's subjective utility for information.
That is,
the amount the market maker is willing to pay out to traders
should reflect the market maker's utility for the information
that the traders have provided.  In this section, we study how the
traders are rewarded for various kinds of information, and use the
magnitude of their profits to define the market maker's implicit
``utility for information'' formally.

We start by defining the market maker's utility for a belief, where a
\emph{belief} $\mub\in\hull$ is a vector of expected
security payoffs $\E[\rhob(\omega)]$ for some distribution over
$\Omega$.


\begin{definition}
The market maker's \emph{utility for a belief} $\mub\in\hull$ relative
to the state $\q$ is the maximum expected
payoff achievable by a trader with belief $\mub$ when the current market
state is $\q$:
\defline{
   \Value(\mub;\q)\coloneqq \sup_{\r \in \reals^K}\bigBracks{\mub\inprod\r-C(\q+\r)+C(\q)}
\enspace.}
\end{definition}

Any subset $\event\subseteq\Omega$ is referred to as an \emph{event}.
Observations $X=x$ correspond to events $\Omega^x$.  Suppose that a
trader has observed an event, i.e., a trader knows that
$\omega\in\event$, but is otherwise uninformed.  The market maker's
utility for that event can then be naturally defined as follows.

\begin{definition}
\label{def:event-value}
The \emph{utility for a (non-null) event}
$\event\subseteq\Omega$
relative
to the market state $\q$ is the largest guaranteed payoff
that a trader who knows $\omega \in \event$ (and has only this information)
can achieve when the current market state is $\q$:
\defline{
   \Value(\event;\q)\coloneqq
   \!\adjustlimits\sup_{\r \in \reals^K}\min_{\omega\in\event}
   \BigBracks{
     \rhob(\omega)\inprod\r - C(\q+\r) + C(\q)
   }
\enspace.}
\end{definition}

Finally, consider the setting in which a trader has observed an event
$\event$, and also holds a belief $\mub$ consistent with
$\event$. Specifically, let $\hull(\event)$ denote the convex hull of
$\set{\rhob(\omega)}_{\omega\in\event}$, which is the set of beliefs
consistent with the event $\event$, and assume
$\mub\in\hull(\event)$. Then we can define the ``excess utility for
the belief $\mub$'' as the excess utility provided by $\mub$ over just
the knowledge of $\event$.

\begin{definition}
\label{def:formalism-cond-value}
Given an event $\event$ and
a belief $\mub\in\hull(\event)$,
the \emph{excess utility of $\mub$ over $\event$}, relative to the
state $\q$ is:
\defline{
  \Value(\mub\given\event;\q)
  =
  \Value(\mub;\q)
  -
  \Value(\event;\q)
\enspace.}
\end{definition}

Note that in these definitions a trader can always choose not to trade
($\r=\0$), so the utility for a belief and an event is
non-negative. Also it is not too difficult to see that
$\Value(\mub;\q)\ge\Value(\event;\q)$ for any $\mub \in \hull(\event)$, so
the excess utility for a belief is also non-negative.



In \Sec{bregman}, we show that given a state $\q$ and a non-null
event $\event$, there always exists a (possibly non-unique)
belief $\mub\in\event$
such that $\Value(\mub\given\event;\q)=0$. Thus, a trader
with such a ``worst-case'' belief is able to achieve
in expectation no reward beyond
what any trader that just observed $\event$ would receive.
We show that these worst-case beliefs correspond to certain kinds of
``projections'' of the current price $\p(\q)$ onto
$\M(\event)$. For LMSR, the projections are with respect
to KL divergence and correspond to the usual conditional
probability distributions. Moreover, for sufficiently smooth
cost functions (including LMSR) they correspond to
market prices that result when a trader is
optimizing his guaranteed profit from the information
$\omega\in\event$ as in \Def{event-value} (see \App{cond:price}).
Because of this motivation,
such beliefs are referred to as ``conditional price vectors.''

\begin{definition} A vector $\mub\in\M(\event)$
is called a \emph{conditional
price vector}, conditioned on $\event$, relative
to the state $\q$ if
$\Value(\mub;\q)=\Value(\event;\q)$. The set
of such conditional price vectors is denoted
\defline{
   \p(\event;\q)
   \coloneqq\set{\mub\in\M(\event):\:\Value(\mub;\q)=\Value(\event;\q)}
\enspace.}
\end{definition}



See Appendix~\ref{app:value} for additional motivation for our definitions of
utility and conditioning. With these notions
defined, we can now state our desiderata.

\subsection{DESIDERATA}
\label{sec:desiderata}

Recall that we aim to design mechanisms
which replace a cost function $C$ at a state $\s$,
with a new cost function $\tC$ at a state $\tsb$.
Let $\Value$ denote the utility for  information with respect
to $C$ and $\tValue$ with respect to $\tC$, and
let $\p$ and $\tpb$ be the respective price maps.
In our mechanisms, we attempt to satisfy
(a subset of) the conditions on
information structures as listed in \Tab{desiderata}.


\begin{table}
\caption{Information Desiderata}%
\label{tab:desiderata}%
\renewcommand{\arraystretch}{1.15}%
\renewcommand{\tabcolsep}{3pt}%
\centering%
\small%
\vspace{-10pt}%
\begin{tabular}{lp{2.4in}}
\toprule
\Price
&
\emph{Preserve prices:}\\&
 $\quad\tpb(\tsb)=\p(\s)$.
\\
\addlinespace
\CondPrice
&
\emph{Preserve conditional prices:}\\&
 $\quad\tpb(X{=}\,x;\tsb)=\p(X{=}\,x;\s)\quad\forall x\in\X$.
\\
\addlinespace
\DecVal
&
\emph{Decrease profits for uninformed traders:}\\&
$\quad\tValue(X{=}\,x;\tsb)\le\Value(X{=}\,x;\s)\quad\forall x\,{\in}\,\X,\!\!\!$\\&
 with sharp inequality if $\Value(X{=}\,x;\s)>0$.
\\
\addlinespace
\ZeroVal
&
\emph{No profits for uninformed traders:}\\&
$\quad\tValue(X{=}\,x;\tsb)=0\quad\forall x\in\X$.
\\
\addlinespace
\CondVal
&
\emph{Preserve excess utility:}\\&
$\quad
   \tValue(\mub{\given} X{=}\,x;\tsb)
   =
   \Value(\mub{\given} X{=}\,x;\s)
$\\&
  for all $x\in\X$ and $\mub\in\hull(X{=}\,x)$.
\\
\bottomrule
\end{tabular}
\figsqueeze
\end{table}

Conditions \Price and \CondPrice capture the requirement to preserve
the information gathered in the market. The current price $\p(\q)$ is
the ultimate information content of the market at a state $\q$
\emph{before} the observation time, but it is not necessarily the
right notion of information content \emph{after} the observation time.
When we do not know the realization $x$, we may wish to set up the
market so that any trader who has observed $X=x$ and would like to
maximize the guaranteed profit would move the market to the same
conditional price vector
as in the previous market. This is captured by \CondPrice.

\ignore{
In other words, we view the
current market state as representing a ``superposition'' of conditional
price vectors, and we require that the new market
preserve this superposition until a trader ``collapses'' it onto the
correct price in some $\hull^x$. For LMSR, this ``superposition''
corresponds to the conditional probability distribution. And the
``collapse'' is the posterior inference (conditioned on the observation).
}

\DecVal models a scenario in which the utility for information about
$X$ decreases over time, and \ZeroVal represents the extreme case in
which utility decreases to zero.  These conditions are in friction
with \CondVal, which aims to maintain the utility structure over the
conditional outcome spaces. A key challenge is to satisfy \CondVal and
\ZeroVal (or \DecVal) simultaneously.

Apart from the information desiderata of \Tab{desiderata},
we would like to maintain an important feature of
cost-function-based market makers:
their ability to bound
the worst-case loss to the market
maker. Specifically, we would like to show that there is some
\emph{finite} bound (possibly depending on the initial state) such
that no matter what trades are executed and which outcome $\omega$
occurs, the market maker will lose no more than the amount of the
bound. It turns out that the solution concepts introduced in this paper maintain the
same loss bound as guaranteed for using just the market's original cost
function $C$, but since the focus of the paper is on the information
structures, worst-case loss analysis is relegated to \App{wcl}.

\miro{Need to discuss how our work relates
to the desiderata of ACV, and also say something like:
Additional types of desiderata have been introduced
in literature to adjust the pricing to the trading
volume~\cite{LV13,OS12,A+14}. That line of work is
orthogonal to ours since we focus on the adaptation
to the information revelation, but we expect that
the two kinds of adjustments can be combined.}



In \Sec{one-step}, we study in detail
the sudden revelation setting
with the goal of instantiating \Prot{switch} in a way that achieves \ZeroVal while satisfying \CondPrice and \CondVal.
Our key result is a characterization and a geometric sufficient condition
for when this is possible.

In \Sec{gradual}, we examine instantiations of \Prot{time-sensitive} for the gradual decrease setting.
Our construction focuses on linearly-constrained
market makers (LCMM) \citep{DudikLaPe12}, which naturally decompose
into submarkets. We show how to achieve \Price, \CondPrice, \DecVal
and \CondVal in LCMMs. We also show that it is possible to
simultaneously decrease the utility for  information in each submarket
according to its own schedule, while maintaining \Price.

Before we develop these mechanisms, we introduce the machinery
of Bregman divergences, which helps us analyze notions
of utility for information.


\subsection{BREGMAN DIVERGENCE AND UTILITY}
\label{sec:bregman}

To analyze the market maker's utility for information,
we show how it corresponds to a specific
notion of distance built into the cost
function, the \emph{mixed (or generalized) Bregman divergence}~\citep{DudikLaPe12,Gordon99}.
Let $R$ be the conjugate of~$C$.~\footnote{The conjugate is also, less commonly, called the ``dual''.}
The mixed Bregman divergence between a belief $\mub$ and a state $\q$ is defined as
$
D(\mub \| \q) \coloneqq R(\mub) + C(\q) - \q \cdot \mub
$.
The conjugacy of $R$ and $C$ implies that
$D(\mub\|\q)\ge 0$
with equality iff $\mub\in\partial C(\q)=\p(\q)$,
i.e.,
if the price vector ``matches'' the state
(see \App{convex}).
The geometric interpretation of mixed Bregman
divergence is as a gap between a tangent and the graph
of the function $R$ (see \Fig{bregman}).

\ignore{
Mixed Bregman divergence has a natural
interpretation within the graph of $R$.
First note
that by conjugacy of~$R$ and~$C$, we can obtain any
non-vertical tangent to the graph of $R$ as a set of points
$(\mub,t(\mub))\in\reals^{K+1}$ where
$\mub\in\reals^K$ and $t(\mub)=\mub\inprod\q-C(\q)$, where
$\q\in\reals^K$ is a predetermined ``slope''
of the tangent.
The divergence $D(\mub\|\q)=R(\mub)-\bigBracks{\mub\inprod\q-C(\q)}$
then measures the gap between the tangent with slope $\q$
and the graph of~$R$.
~\footnote{%
A symmetric interpretation, not
used in this paper, is of course possible within
the graph of $C$.}
}

\begin{figure}
\hspace{35pt}\begin{tikzpicture}[y=0.80pt, x=0.8pt,yscale=-1, inner sep=0pt, outer sep=0pt]
\path[fill=black] (338.56897,224.05708) node[above right] (text2989) {\small $D(\mub\,\|\,\q)$};
\path[fill=black] (267.50586,213.59575) node[above right] (text2993) {\small $R$};
\path[fill=black] (311.86877,271.65604) node[above right] (text2997) {\small $\mub$};
\path[fill=black] (150.81657,260.25974) node[above right] (text3001) {\parbox{168pt}{{\small tangent $t$ with slope $\q$}}};
\path[draw=black,line join=round,line cap=round,miter limit=10.00,line width=0.575pt] (331.4944,171.4539) .. controls (305.3551,204.1240) and (261.4942,228.4980) .. (218.9190,243.5822) .. controls (185.9356,229.5686) and (157.2606,207.2829) .. (142.6293,177.2863);
\path[draw=black,line join=round,line cap=round,miter limit=10.00,line width=0.575pt] (140.3593,235.5850) -- (341.5648,256.7853);
\path[draw=black,dash pattern=on 1.74pt off 1.74pt,line join=round,line cap=round,miter limit=10.00,line width=0.581pt] (316.9530,171.3631) .. controls (316.9530,171.3631) and (316.9530,171.3631) .. (316.9530,171.3631) -- (316.9530,261.8808);
\path[fill=black,nonzero rule] (329.1189,245.3451) -- (329.1189,245.6089) -- (329.1119,245.8867) -- (329.0956,246.1643) -- (329.0711,246.4559) -- (329.0383,246.7475) -- (328.9973,247.0391) -- (328.9481,247.3446) -- (328.8825,247.6500) -- (328.8171,247.9555) -- (328.7433,248.2610) -- (328.6532,248.5664) -- (328.5466,248.8719) -- (328.4401,249.1635) -- (328.3171,249.4690) -- (328.1778,249.7605) -- (328.0304,250.0521) -- (327.8747,250.3298) -- (327.6945,250.5937) -- (327.5141,250.8574) -- (327.3093,251.1213) -- (327.0881,251.3573) -- (326.8587,251.5933) -- (326.6128,251.8016) -- (326.3424,252.0099) -- (326.0638,252.2043) -- (325.7688,252.3709) -- (325.4492,252.5375) -- (325.1214,252.6764) -- (324.9494,252.7320) -- (324.7692,252.7876) -- (324.5888,252.8431) -- (324.4004,252.8847) -- (324.2119,252.9263) -- (324.0153,252.9680) -- (323.8105,252.9957) -- (323.6056,253.0234) .. controls (323.4008,253.0928) and (323.3525,253.1446) .. (323.3525,253.5611) .. controls (326.2625,253.0076) and (327.5365,252.0373) .. (328.3868,250.9019) .. controls (329.5611,249.3335) and (329.6668,247.4499) .. (329.8145,245.9143) -- (329.8145,230.1409) .. controls (329.8145,227.4195) and (329.8145,225.2118) .. (331.4533,222.9347) .. controls (332.8872,220.9213) and (334.3852,220.4150) .. (335.2784,220.3456) .. controls (335.4834,220.2900) and (335.7060,220.3659) .. (335.7060,219.9633) .. controls (335.7060,219.6045) and (335.6800,219.7324) .. (335.2785,219.6630) .. controls (332.5745,219.3852) and (330.3881,216.4226) .. (329.9456,213.0347) .. controls (329.8145,212.2849) and (329.8145,212.1461) .. (329.8145,209.6606) -- (329.8145,195.9423) .. controls (329.8145,193.0542) and (329.8145,190.8187) .. (327.8561,188.1945) .. controls (326.4324,186.7597) and (324.6815,186.4560) .. (323.2450,186.3952) .. controls (323.2450,187.0894) and (323.4908,186.8337) .. (323.9006,186.9032) .. controls (326.4571,187.1531) and (328.4728,189.3886) .. (328.9973,192.9015) .. controls (329.1203,193.5402) and (329.1203,193.6652) .. (329.1203,196.1645) -- (329.1203,210.7021) .. controls (329.1203,213.8678) and (329.4398,215.0341) .. (330.7509,217.2696) .. controls (331.6030,218.6998) and (332.7993,219.3940) .. (333.9383,219.9633) .. controls (330.5952,221.5462) and (329.1203,224.7397) .. (329.1203,228.7525) -- (329.1199,245.3450) -- cycle;

\end{tikzpicture}

\caption{
The mixed Bregman divergence $D(\mub\|\q)$ derived from the
conjugate pair $C$ and $R$ measures the distance
between the tangent with slope $\q$ and
the value of $R$ evaluated at $\mub$. By conjugacy,
the tangent $t$ is described by $t(\mub)=\mub\inprod\q-C(\q)$.
Note that the divergence is well defined even when $R$ is
not differentiable, because each slope vector determines a unique
tangent.}
\label{fig:bregman}
\figsqueeze
\end{figure}
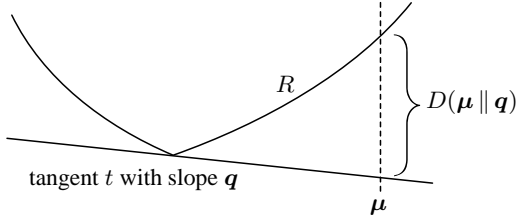

To see how the divergence relates to traders' beliefs,
consider a trader who believes that
$\E[\payb(\omega)] = \mub'$ and moves the market from state $\q$ to state $\q'$.  The expected payoff to this trader is
$(\q'-\q) \cdot \mub' - C(\q') + C(\q)
= D(\mub' \| \q) - D(\mub' \| \q')$.
\ignore{
\begin{align*}
&
(\q'-\q) \cdot \mub' - C(\q') + C(\q)
\\
&\quad{}
= R(\mub') + C(\q) - \q \cdot \mub
- R(\mub') - C(\q') + \q' \cdot \mub
\\
&\quad{}
= D(\mub' \| \q) - D(\mub' \| \q')
\enspace.
\end{align*}
}
This payoff increases as $D(\mub' \| \q')$ decreases. Thus,
subject to the trader's budget constraints, the trader is incentivized
to move to the state $\q'$ which is as ``close'' to his/her belief $\mub'$
as possible in the sense of a smaller value $D(\mub' \| \q')$,
with the largest expected payoff when $D(\mub' \| \q')=0$. This
argument shows that $D(\cdot\|\cdot)$ is an implicit measure of distance
used by traders.

The next theorem shows that the Bregman divergence also matches the
concepts defined in \Sec{val:info}.
Specifically, we show that
(1) the utility for a belief coincides with the Bregman divergence,
(2) the utility for an event $\event$
is the smallest divergence between the current market state
and $\hull(\event)$, and (3) the
conditional price vector
is the (Bregman) projection of the current market state on
$\hull(\event)$, i.e., it is a belief in $\hull(\event)$ that is ``closest to'' the current
market state.

\begin{theorem}
\label{thm:D}
Let $\mub\in\hull$, $\q\in\reals^K$ and $\emptyset\ne\event\subseteq\Omega$.
Then
\begin{align}
&
\label{eq:D:mu}
\Value(\mub;\q) = D(\mub\|\q)
\enspace,
\\
&\textstyle
\label{eq:D:E}
\Value(\event;\q) = \min_{\mub'\in\hull(\event)} D(\mub'\|\q)
\enspace,
\\
&\textstyle
\label{eq:p:E}
\p(\event;\q) = \argmin_{\mub'\in\hull(\event)} D(\mub'\|\q)
\enspace.
\end{align}
\end{theorem}

We finish this section by characterizing when \CondVal
is satisfied and showing that it implies \CondPrice. Recall
that $\Omega^x=\set{\omega:\:X(\omega)=x}$ and let
$\hull^x\coloneqq\hull(\Omega^x)$.

\begin{proposition}
\label{prop:condval}
\CondVal holds if and only if for all $x \in \X$,
there exists some $c^x$ such that for all $\mub\in\hull^x$,
%
 $D(\mub\|\s) - \tD(\mub\|\tsb) = c^x$.
Moreover, \CondVal implies \CondPrice.
\end{proposition}



\section{SUDDEN REVELATION}
\label{sec:one-step}

In this section, we consider the design of sudden revelation market
makers (\Prot{switch}). In this setting, partial information in the
form of the realization of $X$ is revealed to market participants
(but 
not to the market maker) at a predetermined time, as
might be the case if the medal winners of an Olympic event are
announced but no human is available to input this information into
the automated market maker 
on behalf of the market organizer.
The random variable $X$ and the observation time are assumed to be
known, and the market maker wishes to ``close'' the submarket with
respect to $X$ just before the observation time, without knowing the
realization $x$, while leaving the rest of the market unchanged.

Stated in terms of our formalism, we wish to find functions
$\AdvanceState$ and $\AdvanceCost$ from Protocol~\ref{prot:switch}
such that the desiderata \CondPrice, \CondVal, and \ZeroVal from
Table~\ref{tab:desiderata} are satisfied. \jenn{Justify not wanting \Price?}  This implies that traders who know only that $X=x$ are not rewarded after the observation time, but traders with new information about the outcome space conditioned on $X=x$ are rewarded exactly as before.  As a result, trading immediately resumes in a ``conditional market'' on $\hull(\Omega^x)$ for the correct realization $x$, without the market maker needing to know $x$ and without any other human intervention.  We refer to the goal of simultaneously achieving \CondPrice, \CondVal, and \ZeroVal as achieving \emph{implicit submarket closing}.


 For convenience, throughout this section we write $\hull^x \coloneqq \hull(\Omega^x)$ to denote the conditional price space, and $\hull^\star\coloneqq\bigcup_{x\in\X} \hull^x$ to denote prices possible after the observation.

\subsection{SIMPLIFYING THE OBJECTIVE}

We first show that achieving implicit submarket closing can be reduced
to finding a function $\tR$ satisfying a simple set of constraints, and
defining $\AdvanceCost$ to return the conjugate $\tC$ of $\tR$. As a first
step, we observe that it is without loss of generality to let
$\AdvanceState$ be an identity map, i.e., to assume that $\tsb = \s$;
when this is not the case, we can obtain an equivalent market by
setting $\tsb = \s$ and shifting $\tC$ so that the Bregman divergence
is unchanged.
%
%
\begin{lemma}
Any desideratum of \Tab{desiderata} holds for $\tC$ and $\tsb$ if and only if it holds for
$   \tC'(\q)=\tC(\q+\tsb-\s)
\text{ and }
   \tsb'=\s$.
\label{lem:sequalsts}
\end{lemma}

To simplify exposition, we assume that $\tsb = \s$ throughout the rest of the section as we search for conditions on $\AdvanceCost$ that achieve implicit submarket closing.  Under this assumption, Proposition~\ref{prop:condval} can be used to characterize our goal in terms of $\tR$.  Specifically, we show that \CondVal and \CondPrice hold if $\tR$ differs from $R$ by a (possibly different) constant on each conditional price space $\hull^x$.

\begin{lemma}
  \label{lem:one-step-condval}
  When $\tsb=\s$, \CondVal and \CondPrice hold together
  if and only if there exist constants $b^x$ for $x\in\X$ such that
$   \tR(\mub)=R(\mub)-b^x$
  for all  $x\in\X$ and $\mub\in\hull^x$.
%
\end{lemma}

This suggests parameterizing our search for $\tR$
by vectors $\b=\set{b^x}_{x\in\X}$.
For $\b\in\reals^\X$, define a function
\defline{
   R^\b(\mub)
   =
\begin{cases}
   R(\mub)-b^x&\text{if $\mub\in\hull^x, x \in \X$,}
\\
   \infty&\text{otherwise.}
\end{cases}
}
If the sets $\hull^x$ overlap, $R^\b$ is not well defined for all~$\b$. Whenever we write $R^\b$, we assume that $\b$ is such that $R^\b$ is well defined.
To satisfy Lemma~\ref{lem:one-step-condval} with a specific $\b$, it suffices to find a convex function $\tR$ ``consistent with'' $R^\b$ in the following sense.
\begin{definition}
We say that a function $\tR$ is \emph{consistent} with $R^\b$ if
$\tR(\mub)=R^\b(\mub)$ for all $\mub\in \hull^\star$.
\end{definition}




\ignore{
In the next few sections, we show how to simplify our objective further by proving that whenever implicit submarket closing is achievable, it suffices to consider setting $\tC$ to the conjugate of a function $\tR$ of a certain form consistent with $R^\hbb$ for a particular vector $\hbb$, which will be defined below. In particular, we will show that $\tR$ can be obtained by extending $R^\hbb$ beyond $\hull^\star$ by taking the largest convex function consistent with $R^\hbb$ if such a function exists.
}

We next simplify our objective further by proving that whenever implicit submarket closing is achievable, it suffices to consider functions $\AdvanceCost$ that set $\tC$ to be the conjugate of the largest convex function consistent with $R^\b$ for some  $\b\in\reals^\X$.  To establish this, we examine properties of the \emph{convex roof} of $R^\b$, the largest convex function that lower-bounds (but is not necessarily consistent with) $R^\b$.


\ignore{
Lemma~\ref{lem:one-step-condval} restricts
the values of function $\tR$ on the set $\hull^\star$
to coincide with some $R^\b$.
In this section, we show that with these values fixed,
we can extend function $\tR$ beyond $\hull^\star$ by
taking the largest convex function which is consistent
with these values, i.e., which is consistent with $R^\b$,
\emph{whenever such a function exists}.
}


\begin{definition}
  \label{def:roof}
  Given a function $f:\reals^K\to\realsplusinf$,
  the \emph{convex roof} of $f$, denoted $(\conv f)$,  is the largest convex function lower-bounding
  $f$, defined by
  \defline{
  (\conv f)(\x) \coloneqq \sup\Set{g(\x):\:
      g \in \mathcal{G},\,
      g \leq f
  }}
  where $\mathcal{G}$ is the set of convex functions
  $g:\reals^K\to\realsplusinf$, and the condition $g\le f$ holds
  pointwise.
\end{definition}
The convex roof is analogous to a convex hull, and the epigraph of $(\conv f)$ is the convex hull of the epigraph of $f$. See~\citet[{\sectsymb}B.2.5]{urruty2001fundamentals} for details.

\begin{example}
\label{ex:one-step-square}
Recall the square market of \Ex{formalism-square}.
%
%
Let $X(\omegab) = \omega_1$, so traders observe the payoff of the first security at observation time.
Then $\hull^x = \{x\}\times[0,1]$ for $x \in \{0,1\}$.  For
simplicity, let $\b =\0$.  We have
$R^\b (\mub)
  = \mu_2\ln\mu_2 + (1-\mu_2)\ln(1-\mu_2)$ for $\mub \in \hull^1 \cup
  \hull^2$ and $R^\b(\mub) = \infty$ for all other $\mub$.
Examining the convex hull of the epigraph of $R^\b$ gives us that for all $\mub\in[0,1]^2$,
we have $(\conv R^\b) (\mub) =\mu_2\ln\mu_2 + (1-\mu_2)\ln(1-\mu_2)$.
%
\end{example}

\jenn{I cleaned the text slightly to make it less confusing but I
  still think this section desperately cries out for a picture.  Add
  for long version?}  As this example illustrates, the roof of $R^\b$
is the ``flattest'' convex function lower-bounding $R^\b$.
Given the geometric interpretation of Bregman divergence
(\Fig{bregman}), a ``flatter'' $\tR$ yields a smaller utility for information.  This flatness plays a key role in achieving \ZeroVal.
Assume that $\tR$ is consistent with $R^\b$,
so \CondPrice and \CondVal hold by Lemma~\ref{lem:one-step-condval}.
Following the intuition in \Fig{bregman},
to achieve \ZeroVal, i.e., $\tD(\hmub^x\|\s)=0$ across all $x\in\X$
and $\hmub^x\in\p(\Omega^x;\s)$,
it must be the case that for all $x$ and $\hmub^x$,
the function values $\tR(\hmub^x)$ lie on the tangent of $\tR$ with
slope $\s$.
That is, the graph of $\tR$ needs to be \emph{flat} across the points $\hmub^x$. This suggests that the roof might be a good candidate for $\tR$.
This intuition is formalized in the following lemma,
which states that instead of considering arbitrary convex $\tR$ consistent with $R^\b$,
we can consider $\tR$ which take
the form of a convex roof.

\begin{lemma}
  \label{lem:roof-optimality}
  If any convex function $\tR$ is consistent with $R^\b$ then so is the convex roof $\tR'=(\conv R^\b)$. Furthermore,
  if $\tR$ satisfies \ZeroVal or \DecVal then so does $\tR'$.
\end{lemma}

\subsection{IMPLICIT SUBMARKET CLOSING}

We now have the tools to answer the central question of this section: When can we achieve implicit submarket closing?  Lemma~\ref{lem:sequalsts} implies that we can assume that $\AdvanceState$ is the identity function, and Lemmas~\ref{lem:one-step-condval} and~\ref{lem:roof-optimality} imply that it suffices to consider functions $\AdvanceCost$ that set $\tC$ to the conjugate of $\tR=(\conv R^\b)$ for some $\b\in\reals^\X$. What remains is to find the vector $\b$ that guarantees \ZeroVal. As mentioned above, \ZeroVal is satisfied if and only if $\bigParens{\hmub^x,\tR(\hmub^x)}$ lies on the tangent of $\tR$ with the slope $\s$ for all $x \in \X$ and $\hmub^x\in\p(\Omega^x;\s)$. This implies that
$\tR(\hmub^x)=\hmub^x\inprod\s - c$
for all $x$ and $\hmub^x$ and some constant $c$. The specific choice of $c$
does not matter since $\tD$ is unchanged by vertical shifts
of the graph of $\tR$. For convenience, we set $c=C(\s)$, which
makes the tangents of $R$ and $\tR$ with the slope $\s$ coincide. This
and \Lem{one-step-condval}
then yield the choice of $\b = \hbb$, with
\begin{equation}
  \label{eq:one-step-b-0}
  \hb^x\coloneqq R(\hmub^x)+C(\s)-\hmub^x\inprod\s=D(\hmub^x\|\s)
\end{equation}
for all $x$ and any choice of $\hmub^x\in\p(\Omega^x;\s)$.
The resulting construction of $\tR=(\conv R^\hbb)$ can be described using geometric intuition.
First, consider the tangent of $R$ with slope equal to the current
market state $\s$. For each $x \in \X$, take the subgraph of $R$ over
the set $\hull^x$ and let it ``fall'' vertically until it touches this
tangent at the point $\hmub^x$. The set of fallen graphs for all $x$ together describes $R^\hbb$ and the convex hull of the fallen epigraphs yields $\tR=(\conv R^\hbb)$.

\ignore{
\begin{itemize}
\item Take the tangent of $R$ with the slope equal to the current state $\s$.
\item Consider subgraphs of $R$ over sets $\hull^x$ and let them ``fall'' vertically
      until they touch the tangent. Note that they
      will touch the tangent at the points $\hmub^x$.
\item The fallen graphs together describe $R^\hbb$. The convex hull of the fallen epigraphs yields $\tR=(\conv R^\hbb)$.
\end{itemize}
}

\raf{the following sentence seemed like a leap when I read it; maybe we should add justification?}
Defining $\AdvanceCost$ using this construction guarantees \ZeroVal, but \CondPrice and \CondVal are achieved only when $\tR$ is consistent with $R^\hbb$.  Conversely, whenever the three properties are achievable, this construction produces a function $\tR$ consistent with $R^\hbb$. This yields a full characterization of when implicit submarket closing is achievable.

\begin{theorem}
  \label{thm:one-step-0-profits}
Let $\hbb$ be defined as in \Eq{one-step-b-0}.
  \CondPrice, \CondVal, and \ZeroVal can be satisfied using \Prot{switch}
  if and only if $(\conv R^\hbb)$ is consistent with $R^\hbb$. In this case,
  they can be achieved with $\AdvanceState$ as the identity and
  $\AdvanceCost$ outputting the conjugate of $\tR=(\conv R^\hbb)$.
\end{theorem}

\subsection{CONSTRUCTING THE COST FUNCTION}
\label{sec:one-step-examples}

Theorem~\ref{thm:one-step-0-profits} describes how to achieve implicit submarket closing by defining the cost function $\tC$ output by $\AdvanceCost$ implicitly via its conjugate $\tR$.  In this section, we provide an explicit construction of the resulting cost function, and illustrate the construction through examples.

Fixing $R$, for each $x \in \X$ define a function $C^x(\q)\coloneqq \sup_{\mub\in\hull^x} \bigBracks{ \q\inprod\mub - R(\mub) }$. Each function $C^x$ can be viewed as a bounded-loss and arbitrage-free cost function for outcomes in $\Omega^x$.  The conjugate of each $C^x$ coincides with $R$ on $\hull^x$ (and is infinite outside $\hull^x$).  The explicit expression for $\tC$ is described in the following proposition.
\begin{proposition}
  \label{prop:one-step-roof-dual}
  For a given $C$ with conjugate $R$, define $\hbb$ as in \Eq{one-step-b-0} and let $\tR = (\conv R^\hbb)$.  The conjugate $\tC$ of $\tR$ can be written
 $\tC(\q) = \max_{x\in\X} \bigBracks{\hb^x + C^x(\q)}$.
  Furthermore, for each $x \in \X$, $\hb^x=C(\s)-C^x(\s)$.
\end{proposition}

At any market state $\q$ with a unique $\hat{x} \coloneqq \argmax_{x\in\X} \bigBracks{\hb^x + C^x(\q)}$, the price
according to $\tC$ lies in the set $\hull^{\hat{x}}$.  When $\hat{x}$ is not unique, the market has a bid-ask spread.  The addition of $\hb^x$ ensures that the bid-ask spread at the market state $\s$ contains conditional prices $\hmub^x$ across all $x$. \miro{In the longer version:
As a corollary, we have $\tC(\s) = C(\s)$, so the cost function value does not jump at the observation time.}
To illustrate this construction, we return to the example of a square.

\ignore{
To begin the construction of
the cost function corresponding to the roof $(\conv R^\b)$, we first construct specific cost
functions for submarkets
with price spaces $\hull^x=\hull(\Omega^x)$ (i.e., price spaces conditional on $X=x$):
\begin{equation}
  \label{eq:one-step-submarket-cost}
  C^x(\q)\coloneqq \sup_{\mub\in\hull^x} \bigBracks{ \q\inprod\mub - R(\mub) }
\enspace.
\end{equation}
This is analogous to \Eq{conjugate}.
In words, $C^x$ is the bounded-loss and arbitrage-free cost function for outcomes in $\Omega^x$,
whose conjugate coincides with $R$ on $\hull^x$ (and is necessarily infinite
outside). The agreement with $R$ on $\hull^x$ aims to preserve the
original utility for information structure
(following the intuition of \Fig{bregman}).
} 

\ignore{
Using cost functions $C^x$, we can derive a natural expression for the new cost function
$\tC$. It is based on the intuition that we want to prevent just-in-time arbitrage of knowing the
correct $x$. We achieve this by choosing $\tC(\q)=\max_{x\in\X}[b^x+C^x(\q)]$ for some $b^x$.
The effect of taking the maximum is that we choose prices from $\hull^x$ when the maximum is at a unique $x$,
or create a bid-ask spread if $x$ is non-unique. Values of $b^x$ are selected to ensure
that there is a full bid-ask spread at the state $\tsb=\s$. It turns out that this is
achieved by $\hb^x$:
} 

\begin{example}
  \label{ex:one-step-square-2}
Consider again the square market from Examples~\ref{ex:formalism-square} and \ref{ex:one-step-square} with $X(\omegab) = \omega_1$.
One can verify that $C^x(\q) = x q_1 + \ln\bigParens{1+e^{q_2}}$
for $x\in\set{0,1}$. \Prop{one-step-roof-dual} gives
\\[4pt]$\tC(\q)
  = \max_{x\in\{0,1\}} \BigBracks{
        x (q_1 - s_1) + \ln(1+e^{q_2}) + \ln(1+e^{s_1})
     }$
\\$\phantom{\tC(\q)}= \max\{0,q_1-s_1\} + \ln(1+e^{s_1}) + \ln(1+e^{q_2}).$
\\[6pt]
In switching from $C$ to $\tC$ we have effectively changed the first
term of our cost from a basic LMSR cost for a single binary security
to the piecewise linear cost of
Example~\ref{ex:formalism-piecewise-linear}, introducing a bid-ask spread for security 1 when $q_1=s_1$; states $\q = (s_1,q_2)$ have $\tpb(\q)=[0,1]\times \{e^{q_2}/(1+e^{q_2})\}$.  The market for security 1 has thus implicitly closed; as the new market begins with $\q=\s$, any trader can switch the price of security 1 to 0 or 1 by simply purchasing an infinitesimal quantity of security 1 in the appropriate direction, at essentially no cost and with no ability to profit.
\end{example}

The example above illustrates our cost function construction, but does not show that $\tR$ is consistent with $R^\hbb$
as required by Theorem~\ref{thm:one-step-0-profits}.  In fact, it is consistent.
This follows from the sufficient condition proved in \App{extend-R-b}.
Briefly, the condition is that $\hull^\star$ does not contain any price vectors $\mub$ that can be expressed as nontrivial convex combinations of vectors from multiple $\hull^x$.

In \App{suff:examples}, we show that this sufficient condition applies
to many settings of interest such as arbitrary partitions
of simplex and submarket observations in
binary-payoff LCMMs (defined in \Sec{gradual}),
which were used to run a combinatorial market
for the 2012 U.S. Elections~\citep{DudikEtAl13}.

A case in which the sufficient condition is violated
is the square market with $X(\omegab) = \omega_1 + \omega_2 \in
\{0,1,2\}$, where $\hull^0 = (0,0)$ and $\hull^2 = (1,1)$ but $(\tfrac
1 2,\tfrac 1 2) = \tfrac 1 2 (0,0) +  \tfrac 1 2 (1,1) \in \hull^1$.
This particular
example also fails to satisfy Theorem~\ref{thm:one-step-0-profits}
(see \App{closing:impossible}),
but in general the sufficient condition is not necessary
(see Appendix~\ref{app:ex-roof-nec}).



\section{GRADUAL DECREASE}
\label{sec:gradual}

We now consider gradual decrease market makers
(\Prot{time-sensitive}) for the gradual decrease setting in which the
utility of information about a future observation $X$ is decreasing
continuously over time.  We focus on \emph{linearly
  constrained market makers} (LCMMs)~\cite{DudikLaPe12}, which
naturally decompose into submarkets.
%
%
Our proposed gradual decrease market maker employs a different LCMM
at each time step, and satisfies various desiderata of
\Sec{desiderata} between steps.

As a warm-up for the concepts introduced in this section,
we show how the ``liquidity parameter'' can be used to implement
a decreasing utility for information.

\begin{example}\emph{Homogeneous decrease in utility for information.}
\label{ex:decrease:0}
We begin with a differentiable
cost function
$C$ in a state $\s$. Let $\alpha\in(0,1)$, and
define $\tC(\q)=\alpha C(\q/\alpha)$, and $\tsb=\alpha\s$.
$\tC$ is parameterized by the ``liquidity parameter'' $\alpha$.
The transformation $\tsb$ guarantees the preservation of
prices, i.e.,
$  \tpb(\tsb)=\nabla\tC(\tsb)=\alpha \nabla C(\tsb/\alpha) / \alpha
  = \nabla C(\s)
  = \p(\s)$.
We can derive that
$\tR(\mub)=\alpha R(\mub)$, and
$\tD(\mub\|\q)=\alpha D(\mub\|\q/\alpha)$,
so, for all $\mub$,
$  \tD(\mub\|\tsb)=\alpha D(\mub\|\s)$.
In words, the utility for all beliefs $\mub$
with respect to the current state is
decreased according to the multiplier $\alpha$.
\end{example}

This idea
will be the basis of our construction.
We next define the components of our setup
and prove the desiderata.

\subsection{LINEARLY CONSTRAINED MARKETS}
\label{sec:linearly:constrained}

Recall that $\rhob:\Omega\to\reals^K$ is the payoff function. Let $\G$ be a system of non-empty
disjoint subsets $g\subseteq[K]$ forming a partition of coordinates of $\rhob$,
so $[K]=\bigcup_{g\in\G} g$. We use the notation
$\rhob_g(\omega)\coloneqq\Parens{\rho_i(\omega)}_{i\in g}$ for the
block of coordinates in $g$, and similarly $\mub_g$ and $\q_g$. Blocks
$g$ describe groups of securities that are treated as separate ``submarkets,''
but there can be logical dependencies among them.

\begin{example}\emph{Medal counts.}
  Consider a prediction market for the Olympics. Assume that \Slovakia
  takes part in $n$ Olympic events. In each, \Slovakia can win
  a gold medal or not. Encode this outcome space as
  $\Omega=\set{0,1}^n$.  Define random variables
  $X_i(\omegab)=\omega_i$ equal to 1 iff \Slovakia wins gold in the
  $i$th Olympic event. Also define a random variable $Y=\sum_{i=1}^n
  X_i$ representing the number of gold medals that \Slovakia wins in
  total. We create $K=2n+1$ securities, corresponding to 0/1 indicators of
  the form $\1[X_i=1]$ for $i\in[n]$ and $\1[Y=y]$ for $y\in\set{0,\dotsc,n}$.
  That is, $\rho_i=X_i$ for $i\in[n]$ and $\rho_{n+1+y}=\1[Y=y]$ for
  $y\in\set{0,\dotsc,n}$.  A natural
  block structure in this market is $ \G=\bigl\{
  \set{1},\,\set{2},\dotsc,\set{n},\,\set{n+1,\dotsc,2n+1} \bigr\}$
  with submarkets corresponding to the $X_i$ and $Y$.
\end{example}

Given the block structure $\G$,
the construction of a linearly constrained market begins with bounded-loss
and arbitrage-free convex cost functions $C_g:\reals^g\to\reals$ with
conjugates $R_g$ and divergences $D_g$ for each
$g \in \G$. These cost functions are assumed to be easy to compute
and give rise to a ``direct-sum'' cost
$  \Cx(\q)=\sum_{g\in\G} C_g(\q_g)$
with the conjugate $\Rx(\mub)=\sum_{g\in\G} R_g(\mub_g)$
and divergence $\Dx(\mub\|\q)=\sum_{g\in\G} D_g(\mub_g\|\q_g)$.

Since $\Cx$ decomposes, it can be calculated quickly.
However, the market maker $\Cx$ might
allow arbitrage due to
the lack of consistency among submarkets
since arbitrage opportunities
arise when prices fall outside $\hull$~\cite{Abernethy11}.
$\hull$ is always polyhedral, so it can be described as $\hull =
\Set{\mub\in\reals^K:\:\A^\top\mub\ge\b}$ for some matrix
$\A\in\reals^{K\times M}$ and vector $\b\in\reals^M$.
%
%
Letting $\a_m$ denote the $m$th column of~$\A$, arbitrage
opportunities open up if the price of the bundle $\a_m$ falls below
$b_m$. For any $\etab \in \reals^M_+$, the bundle $\A\etab$ presents
an arbitrage opportunity if priced below $\b \inprod \etab$.

A \emph{linearly constrained market maker} (LCMM) is
described by the cost function
%
 $ C(\q) = \inf_{\etab\in\reals^M_+}\bigBracks{\Cx(\q+\A\etab)-\b\inprod\etab}$.
%
While the definition of $C$ is slightly involved,
the conjugate $R$ has a natural
meaning as a restriction of the direct-sum market to the price space
$\hull$, i.e., $R(\mub)=\Rx(\mub) + \ones\Bracks{\mub\in\hull}$. Furthermore,
the infimum in the definition of $C$ is always attained
(see \App{lcmm}).
Fixing $\q$ and letting $\etab^\star$ be a minimizer in the definition,
we can think of the market maker as automatically charging traders for
the bundle $\A\etab^\star$, which would present an arbitrage
opportunity, and returning to them the guaranteed payout
$\b\inprod\etab$.  This benefits traders while maintaining the same
worst-case loss guarantee for the market maker as $\Cx$
\citep{DudikLaPe12}.

\ignore{
\footnote{\citet{DudikLaPe12} take
  advantage of the fact that when evaluating the cost, it is not
  necessary to optimize $\etab$ exactly.
\ignore{
  In fact, it is possible to
  alternate between two regimes: (i) executing trades while keeping
  $\etab$ fixed (and hence obtaining rapid pricing according to
  $\Cx$), and (ii) optimizing over $\etab$ (possibly imperfectly, to
  remove arbitrage).}
  Also, it is frequently more tractable to just require
  that the linear constraints describe a superset of $\hull$, leaving
  some computationally hard arbitrage to the traders. Here, we for
  simplicity assume that the constraints describe the
  convex hull exactly 
  and the cost function
  calculation follows \Eq{C:lcmm} exactly.}
} 

\ignore{
LCMM can be interpreted as implementing the ``arbitrager'' action in
the market $\Cx$ via the optimal choice of $\etab$. In particular, the
market maker can be viewed as having purchased $\eta_m$ shares of each
bundle $\a_m$, and guaranteed payoff of $\b \inprod \etab$ from these
securities is returned to traders. See \citet{DudikLaPe12} for a
thorough analysis of such markets.
}

\begin{example}\emph{LCMM for medal counts.}
Continuing the previous example, for submarkets
$X_i$, we can define LMSR costs
$C_i(q_i)=\ln\Parens{1+\exp(q_i)}$. For the submarket
for~$Y$, let $g=\set{n+1,\dotsc,2n+1}$ and use the LMSR
cost $C_g(\q_g)=\ln\bigParens{\sum_{y=0}^n \exp(q_{n+1+y})}$.
The submarkets for $X_i$ and $Y$ are linked. One example of a
linear constraint is based on the linearity of expectations:
for any distribution, we must have
$\E[Y]=\sum_{i=1}^n\E[X_i]$. This places an equality
constraint
$\sum_{y=0}^n\; y\cdot\mu_{n+1+y}\;=\;\sum_{i=1}^n\;\mu_i$
on the vector $\mub$, which can be expressed as two inequality constraints
(see \citet{DudikLaPe12,DudikEtAl13} for more
on constraint generation).
\end{example}

\subsection{DECREASING LIQUIDITY}

We now study the gradual decrease scenario in which the utility for information in each
submarket $g$ 
decreases over time. In the Olympics example, the market
maker may want to continuously
decrease the rewards for information about a particular event as the
event takes place.

We generalize the strategy from \Ex{decrease:0}
to LCMMs and extend them to time-sensitive cost functions
by introducing the ``information-utility schedule'' in the form of a differentiable
non-increasing
function $\beta_g:\reals\to(0,1]$ with $\beta_g(t^0)=1$.
The speed of decrease of $\beta_g$
controls the speed of decrease of the utility for information
in each submarket. (We make this statement more precise in
\Thm{time-sensitive}.)

\ignore{
We implement the gradual decrease by extending LCMMs to gradual decrease markets
as in \Prot{time-sensitive}.
In addition to the structure described
in the previous section, our extension requires that
each submarket be provided with the
``information-utility schedule'' in the form of a differentiable
non-increasing
function $\beta_g:\reals\to(0,1]$ with $\beta_g(t^0)=1$.
The speed of decrease of $\beta_g$
controls the speed of decrease of the utility for information
in each submarket. (We make this statement more precise in
\Thm{time-sensitive}.

We adopt the strategy used in \Ex{decrease:0}.}

We first define a gradual decrease direct-sum cost function
  $  \CCx(\q;t) = \sum_{g\in\G} \beta_g(t) C_g\bigParens{\q_g / \beta_g(t)}$
which is used to define a gradual decrease LCMM, and a matching
$\AdvanceState$ as follows:
\begin{align*}
&
\textstyle
   \C(\q;t) = \inf_{\etab\in\reals^M_+}\bigBracks{
         \CCx(\q+\A\etab;t) - \b\inprod\etab}
\\
&
   \AdvanceState(\q;\,t,\ttt) = \tqb
\\
&
   \quad\text{such that }\tqb_g = \tfrac{\beta_g(\ttt)}{\beta_g(t)}(\q_g+\deltab^\star_g) - \deltab^\star_g
\\
&  \quad\text{where $\etab^\star$ is a minimizer in $\C(\q;t)$
                     and $\deltab^\star=\A\etab^\star$}
\enspace.
\end{align*}
When considering the state update from time~$t$ to time~$\ttt$,
the ratio $\beta_g(\ttt)/\beta_g(t)$
has the role of the liquidity parameter $\alpha$ in \Ex{decrease:0}.
The motivation behind the definition of $\AdvanceState$ is to guarantee
that $\tqb_g+\deltab^\star_g=[\beta_g(\ttt)/\beta_g(t)](\q_g+\deltab^\star_g)$,
which turns out to ensure that $\etab^\star$ remains the
minimizer and the prices are unchanged.
The preservation of prices (\Price) is achieved by a scaling similar to \Ex{decrease:0},
albeit applied to the market state in the direct-sum market underlying the LCMM.

This intuition is formalized in the next theorem,
which shows that the above construction
preserves prices and decreases the utility for information, as captured
by the mixed Bregman divergence, according to the schedules
$\beta_g$. We use the notation
$C^t(\q)\coloneqq\C(\q;t)$ and write $D^t_g$
for the divergence derived from $C^t_g(\q_g) \coloneqq \beta_g(t) C_g (\q_g
/ \beta_g(t))$.
\begin{theorem}
\label{thm:time-sensitive}
Let $\C$ be a gradual decrease LCMM, let $t,\ttt\in\reals$ and
$\s\in\reals^K$.
The
replacement of $C^t$ by $\ttC\coloneqq C^{\ttt}$ and $\s$ by
$\tsb\coloneqq\AdvanceState(\s;\,t,\ttt)$ satisfies
\Price. Also,
\begin{equation}
\label{eq:time-sensitive}
   \ttD(\mub\|\tsb) =
   \sum_{g\in\G} \talpha_g D^t_g(\mub_g\|\s_g+\deltab^\star_g)
   + (\A^\top\mub-\b)\cdot\etab^\star
\end{equation}
for all $\mub\in\hull$,
where $\etab^\star$ and $\deltab^\star$ are defined by $\AdvanceState(\s;\,t,\ttt)$,
and $\talpha_g=\beta_g(\ttt)/\beta_g(t)>0$.
\end{theorem}
%
The first term on the right-hand side of  \Eq{time-sensitive} is the sum of divergences in submarkets $g$,
each weighted by a coefficient $\talpha_g$ which is equal to one
at $\ttt=t$ and weakly decreases as $\ttt$ grows.
The divergences are
between $\mub_g$ and the state resulting from the arbitrager action
in the direct-sum market. The second term
is non-negative,
since $\mub\in\hull$, and represents expected arbitrager
gains beyond the guaranteed profit from the arbitrage in the direct-sum
market.
The only terms
that depend on time~$\ttt$ are the
multipliers $\talpha_g$. Since they are decreasing
over time,
we immediately obtain that
the utility for information, $\Value(\mub;\tsb) = \ttD(\mub\|\tsb)$,
is also decreasing, with the contributions
from individual submarkets
decreasing according to their schedules $\beta_g$.

\ignore{
Differentiating the result of \Thm{time-sensitive} with respect to
$\ttt$, we immediately obtain the following corollary showing how the
rate of decrease of the utility for information is determined by the rate
of decrease of the schedules $\beta_g$. \jenn{This corollary may be a
  good candidate to cut or move to the appendix if we need space.}
\begin{corollary}
\label{cor:decrease}
In the setup of \Thm{time-sensitive}, fix $t$ and consider
$\ttt$ as a variable (with $\tsb$ and $\ttD$ varying accordingly).
Then, in the absence of trades, the utility for a belief $\mub\in\hull$
is decreasing according to
\[
  \frac{\partial\ttD(\mub\|\tsb)}{\partial\ttt}
  =
  \sum_g
  \frac{\partial\beta_g(\ttt)}{\partial\ttt}
  \cdot
  \frac{D^t_g(\mub_g\|\s_g+\deltab^\star_g)}{\beta_g(t)}
\enspace.
\]
\end{corollary}

\ignore{
The fact that the utility information in a gradual decrease
LCMM is non-increasing can be used to show that
the loss suffered by the market-maker has a finite
bound independent of the trading activity
(see \App{wcl:gradual}).}
}

%

When only one of the schedules $\beta_g$ is
decreasing and the other schedules stay constant, we can show
that the excess utility and conditional prices
are preserved (conditioned on $\rhob_g$), and under certain conditions
also \DecVal holds.

For a submarket $g$, let $\X_g\coloneqq\Set{\rhob_g(\omega):\:\omega\in\Omega}$
be the set of realizations of $\rhob_g$.
Recall that $\hull(\event)$ is the convex hull of $\set{\rhob(\omega)}_{\omega\in\event}$.
We show that \DecVal holds if $C_g$
is differentiable and the submarket $g$ is ``tight'' as follows.
\begin{definition}
We say that a submarket $g$ is \emph{tight} if for all $\x\in\X_g$
the set $\set{\mub\in\hull:\:\mub_g=\x}$ coincides with
$\hull(\rhob_g=\x)$,
i.e., if all the beliefs $\mub$ with $\mub_g=\x$
can be realized by probability distributions over states $\omega$ with
$\rhob_g(\omega)=\x$. (In general, the former is always a superset
of the latter, hence the name ``tight'' when the equality holds.)
\end{definition}
While this condition is somewhat restrictive, it is easy to see that
all submarkets with binary securities, i.e., with
$\rhob_g(\omega)\in\set{0,1}^g$, are tight
(see \App{tight}).
%
%
\begin{theorem}
\label{thm:gradual:partial}
Assume the setup of \Thm{time-sensitive}. Let $g\in\G$
and assume that
$\beta_g(\ttt)<\beta_g(t)$ whereas
$\beta_{g'}(\ttt)=\beta_{g'}(t)$ for $g'\ne g$.
Then the replacement
of $C^t$ by $\ttC$ and $\s$ by $\tsb$ satisfies
\CondPrice and \CondVal for the random variable
$\rhob_g$. Furthermore, if $C_g$ is differentiable and
the submarket $g$ is tight, we also obtain \DecVal.
\end{theorem}

\jenn{In the long version we could mention that the utility is
  non-increasing even if the tightness condition doesn't hold, so
  nothing too crazy happens in that case.}




\ifnum\Spacehack=1
\let\section\oldsection
\let\subsection\oldsubsection
\fi

\newpage

\bibliographystyle{plainnat}
{\small
\bibliography{seq_mm}}

\newpage
\appendix
\section{CONVEX ANALYSIS}
\label{app:convex}

Here we briefly review concepts and results from convex analysis which we use throughout the paper.

\paragraph{Convex sets, polytopes, relative interior.}

Let $S\subseteq\reals^n$. We say that
$S$ is \emph{convex} if it contains all line segments with
endpoints in $S$. The \emph{convex hull} of $S$, denoted $\conv S$, is the smallest convex
set containing $S$. It
can be characterized as the set containing all ``convex combinations''
of points in $S$~\citep[][Theorem 2.3]{Rockafellar70}, where a \emph{convex combination}
of points $\u_1,\dotsc,\u_k$ is a point
$\u=\sum_{i=1}^k \lambda_i\u_i$
for any $\lambda_i\ge 0$ with $\sum_{i=1}^k\lambda_i=1$.

A set which is a convex hull of a finite set of points is called a \emph{polytope}.
We say that $S$ is \emph{polyhedral} if it is an
intersection of a finite set of half-spaces, i.e., if
$S=\set{\u\in\reals^n:\:\A\u\ge\b}$ for some matrix $\A\in\reals^{m\times n}$
and vector $\b\in\reals^m$. All polytopes are polyhedral~\citep[][Theorem 19.1]{Rockafellar70}.

An \emph{affine hull} of $S$ is the smallest affine space
containing~$S$. The topological interior of $S$ relative to its affine
hull is called the \emph{relative interior} of $S$ and denoted
$\relint S$.  To give a common example, if $S$ is a simplex in $n$
dimensions, i.e.,
$S = \set{\u \in \reals^n: u_i \geq 0,\,\sum_{i=1}^n u_i = 1}$,
then the interior of $S$ is empty, but
$\relint S = \set{\u \in \reals^n: {u_i>0},\,\sum_{i=1}^n u_i = 1}$.

\paragraph{Function properties, epigraph, closure, roof.}

Consider a function $f:\reals^n\to(-\infty,\infty]$. Its \emph{domain},
denoted $\dom f$, is the set of points $\u$ such that $f(\u)$ is finite.
The function $f$ is called \emph{proper} if its domain is non-empty.
The \emph{epigraph} of $f$, denoted $\epi f$, is the set
of points on and above the graph of $f$, i.e.,
\[
  \epi f \coloneqq \Set{ (\u,t)\in \reals^n\times\reals:\: t \geq f(\u) }
\enspace.
\]
The function $f$ is called \emph{closed} if its epigraph is a closed set.
This is equivalent to $f$ being lower semi-continuous~\citep[][Theorem 7.1]{Rockafellar70}.
The function $f$ is called \emph{convex} if its epigraph is
a convex set, or equivalently, if for all $\u, \u' \in \dom f$, for all
$\lambda \in (0,1)$,
\[
f(\lambda \u + (1-\lambda) \u') \leq \lambda
f(\u) + (1-\lambda) f(\u')
\enspace.
\]
The function $f$ is \emph{strictly convex} if the inequality above is strict
whenever $\u \neq \u'$.
Closed convex functions are not only lower semi-continuous,
but actually continuous relative to any polyhedral subset of their domain
(see Theorems~10.2 and 20.5 of \citet{Rockafellar70}).
\begin{proposition}
\label{prop:cont}
Let $f:\reals^n\to(-\infty,\infty]$ be a closed convex function and
$S$ any polyhedral subset of $\dom f$. Then $S$ is continuous relative to $S$.
\end{proposition}
Any convex function finite on all of $\reals^n$
is necessarily continuous~\citep[][Corollary 10.1.1]{Rockafellar70} and therefore
closed.
The \emph{closure} of $f$, denoted $\cl f$,
is the unique function whose epigraph is the topological
closure of $\epi f$. As defined in \Sec{one-step},
the \emph{convex roof} of $f$,
denoted $(\conv f)$,
is the unique function whose epigraph is the convex
hull of $\epi f$.

\paragraph{Subdifferential, conjugacy, duality.}

Consider a convex function $f:\reals^n\to\realsplusinf$.
A \emph{subgradient} of $f$ at a point $\u\in\dom f$ is a vector $\v\in\reals^n$ such that
\[
  f(\u') \geq f(\u) + \v\inprod(\u'-\u)
\]
for all $\u'$. The set of all subgradients of $f$
at $\u$ is called the \emph{subdifferential} of $f$ at $\u$ and
denoted $\partial f(\u)$.
If $f$ is differentiable at $\u$, then $\partial f(\u)$ is the singleton
equal to the gradient of $f$ at $\u$.

Let $f:\reals^n\to\realsplusinf$ be any proper function.
The \emph{(convex) conjugate} of $f$ is the function $f^*:\reals^n\to\realsplusinf$ defined by
\begin{equation}
\label{eq:conj}
  f^*(\v)\coloneqq\sup_{\u\in \reals^n}\Bracks{\v\inprod\u - f(\u)}
\enspace.
\end{equation}
The function $f^*$ is always closed and convex (because its epigraph is
an intersection of half-spaces).
We write $f^{**} = (f^*)^*$ to denote the \emph{biconjugate} of $f$.
The biconjugate is a closure of the convex roof of $f$~\citep[][Theorem E.1.3.5]{hiriart1996convex}.

\begin{proposition}
  \label{prop:appendix-double-conj}
  Let $f:\reals^n\to\realsplusinf$ be a proper convex function or a proper function bounded below by
  an affine function. Then $f^{**} = \cl(\conv f)$. Hence, if $f$ is
  a closed proper convex function, $f^{**} = f$.
\end{proposition}

The definition of the conjugate implies that
\begin{equation}
\label{eq:fenchel}
  f^*(\v)\ge \v\inprod\u-f(\u)
\end{equation}
for all $\u$ and $\v$, with the equality if and only if $\u$ is the maximizer
on the right-hand side of \Eq{conj}. If $f$ is convex, this can only
happen if $\v\in\partial f(\u)$. Similar reasoning can be applied to $f^{**}$,
yielding the following proposition (based on Theorem 23.5 of \citet{Rockafellar70}).
Instead of $f$ and $f^*$, we use the notation $C$ and $R$ to reflect the intended use
in the body of the paper. The gap between the left-hand side and
the right-hand side of \Eq{fenchel} is referred to as the \emph{mixed
Bregman divergence}.

\begin{proposition}
\label{prop:first:order}
Let $C$ be a closed proper convex function, $R$ its conjugate, and $D$
the associated mixed Bregman divergence
$D(\mub\|\q)\coloneqq R(\mub)+C(\q)-\mub\inprod\q$. Then $D(\mub\|\q)\ge 0$
for all $\mub$, $\q$ and the following statements are equivalent:
\squishlist
\item $D(\mub\|\q)=0$
\item $\q\in\partial R(\mub)$
\item $\mub\in\partial C(\q)$
\squishend
\end{proposition}

A function
is called \emph{polyhedral} if its epigraph is polyhedral.
The following theorem relates a convex minimization
problem with a concave maximization problem via convex
conjugates. It is a version of \emph{Fenchel's duality}
and a subcase of Corollary~31.2.1 of~\citet{Rockafellar70}.

\begin{theorem}[Fenchel's duality]
\label{thm:fenchel:duality}
Let $f:\reals^K\to\realsplusinf$ and $g:\reals^M\to\realsplusinf$ be closed convex functions and $\A\in\reals^{K\times M}$. Further assume
that $g$ is polyhedral
and there exists $\mub\in\relint(\dom f^*)$ such that
$\A^\top\mub\in\dom g^*$. Then
\[
   \inf_{\etab\in\reals^M}
   \Bracks{f(\A\etab) + g(\etab)}
   =
   \sup_{\mub\in\reals^K}
   \Bracks{-f^*(\mub) - g^*(- \A^\top\mub)}
\]
and the infimum is attained.
\end{theorem}


\section{PROOFS FROM \SEC{bregman}}
\label{app:bregman}

\subsection{PROOF OF \THM{D}}
\label{app:D}

First we prove \Eq{D:mu} using the definition of $\Value(\mub;\q)$ and
the conjugacy of $R$ and $C$:
\begin{align*}
&
 \Value(\mub;\q)
 =
 \textstyle
  \sup_{\r\in\reals^K}\bigBracks{\mub\inprod\r-C(\q+\r)+C(\q)}
\\
&\quad{}=
\textstyle
\sup_{\r\in\reals^K}\bigBracks{\mub\inprod(\q+\r)-C(\q+\r) - \mub\inprod\q + C(\q)}
\\
&\quad{}
 =R(\mub)-\mub\inprod\q+C(\q)
 =D(\mub\|\q)
\enspace.
\end{align*}
Next, we prove \Eq{D:E}:
\begin{align}
&
\Value(\event;\q)
=
  \adjustlimits\sup_{\r \in \reals^K} \min_{\omega\in\event}
  \BigBracks{
     \r\cdot\rhob(\omega) + C(\q) - C(\q+\r)
  }
\nonumber \\
&\quad{}
=\!\!
  \adjustlimits \sup_{\q' \in \reals^K} \min_{\mub'\in\hull(\event) }
  \BigBracks{
     (\q'-\q)\cdot\mub' + C(\q) - C(\q')
  }\!\!
\label{eq:E:1} \\
&\quad{}
=\!\!
  \adjustlimits \min_{\mub'\in\hull(\event)} \sup_{\q' \in \reals^K}
  \BigBracks{
     (\q'-\q)\cdot\mub' + C(\q) - C(\q')
  }\!\!
\label{eq:E:2} \\
&\quad{}
=\!\!
  \min_{\mub'\in\hull(\event)} \BigBracks{
     R(\mub') -\q\cdot\mub' + C(\q)
  }
\label{eq:E:3}
\\
&\quad{}
 = \min_{\mub'\in\hull(\event)} D(\mub' \| \q)
\enspace,
\notag
\end{align}
where the equalities are justified as follows.
\Eq{E:1} follows by relaxing, without loss of generality,
the optimization of a linear function over $\set{\rhob(\omega)}_{\omega\in\event}$
to the optimization over the convex hull, and substituting $\q'=\q+\r$.
\Eq{E:2} follows from Sion's minimax theorem, and finally \Eq{E:3}
follows from the definition of convex conjugacy.

The final statement to prove, \Eq{p:E}, follows immediately from the definition
of $\p(\event;\q)$ and Eqs.~\eqref{eq:D:mu} and~\eqref{eq:D:E}.
%

\subsection{PROOF OF \PROP{condval}}
\label{app:condval}

It suffices to show that the statement
of the proposition holds
for a specific $x\in\X$ with \CondVal and
\CondPrice also restricted to a specific $x$.
The proposition will then
follow by universal quantification across all $x\in\X$.
Thus in the remainder we consider a specific $x\in\X$.

Let
$\hmub^x\in\p(X=x;\s)$ and
$\tmub^x\in\tpb(X=x;\tsb)$.
The definition of the excess utility for a belief and \Thm{D} then
imply that \CondVal (restricted to $x$) is satisfied if and only if
for all $\mub\in\hull^x$
\begin{equation}
\label{eq:condval:1}
  D(\mub\|\s)-D(\hmub^x\|\s) = \tD(\mub\|\tsb)-\tD(\tmub^x\|\tsb)
\enspace.
\end{equation}

First assume that \CondVal holds and therefore \Eq{condval:1} holds
for $x$. Then the desired condition follows by setting
$c^x=D(\hmub^x\|\s)-\tD(\tmub^x\|\tsb)$.

Conversely, assume that $D(\mub\|\s) - \tD(\mub\|\tsb) = c^x$ holds for
all $\mub \in \hull^x$. Since
$D(\mub\|\s) = \tD(\mub\|\tsb) + c^x$, we obtain
\begin{align*}
  \p(X=x;\s)
  &=
  \argmin_{\mub\in\hull^x} D(\mub\|\s)
\\
  &=
  \argmin_{\mub\in\hull^x} \tD(\mub\|\tsb)
  =
  \tpb(X=x;\tsb)
\enspace,
\end{align*}
i.e., \CondPrice (restricted to $x$) holds.  This and the argument
above show that \CondVal implies \CondPrice.

To finish the proof
we have to show that the assumption that $D(\mub\|\s) - \tD(\mub\|\tsb) = c^x$
also implies \CondVal. Let $\hmub^x\in\p(X=x;\s)=\tpb(X=x;\tsb)$.
Then we have
\[
  D(\mub\|\s) - \tD(\mub\|\tsb)
  = c^x
  = D(\hmub^x\|\s) - \tD(\hmub^x\|\tsb)
\]
and rearranging yields \Eq{condval:1}, with $\hmub^x$ substituted
for $\tmub^x$. However, $\hmub^x$ is a valid choice of $\tmub^x$ since
$\tmub^x$ was chosen arbitrarily from $\tpb(X=x;\tsb)  = \p(X=x;\s)$,
so \Eq{condval:1} and therefore \CondVal hold.
%

\section{PROOFS FROM \SEC{one-step}}
\label{app:one-step}

\subsection{PROOF OF \LEM{sequalsts}}
\label{app:sequalsts}

Theorem~\ref{thm:D} shows that all the desiderata, except for \Price, are derived from properties of $\tD(\mub\|\tsb)$ as a function of $\mub$.
To see that \Price can also be derived this way, note that by \Prop{first:order} we have $\tpb(\tsb)=\partial\tC(\tsb)=\set{\mub:\:\tD(\mub\|\tsb)=0}$. Thus,
it suffices to analyze $\tD$.
With $\tC'$ and $\tsb'$ as in the lemma,
we have $\tR'(\mub)=\tR(\mub)-(\tsb-\s)\inprod\mub$ and
\[
  \tD'(\mub\|\s)
  =\tR(\mub)-(\tsb-\s)\inprod\mub + \tC(\tsb) - \mub\inprod\s
  =\tD(\mub\|\tsb).
\]
Hence the lemma holds.

\subsection{PROOF OF \LEM{one-step-condval}}
\label{app:one-step-condval}

\Prop{condval} shows that \CondVal and \CondPrice are together satisfied if and only if
there exist constants $c^x$ such that for all $x\in\X$ and $\mub\in\hull^x$,
\[
  c^x = D(\mub\|\s) - \tD(\mub\|\s) = C(\s) + R(\mub) - \tC(\s) - \tR(\mub)
\enspace.
\]
If this statement holds, then for any $x$, setting $b^x=c^x-C(\s)+\tC(\s)$ gives us
$\tR(\mub) =  R(\mub) - b^x$ for all $\mub\in\hull^x$.  Conversely,
if $\tR(\mub) =  R(\mub) - b^x$ for all $\mub\in\hull^x$, setting
$c^x = b^x+C(\s)-\tC(\s)$ gives the equation above.

\subsection{PROOF OF \LEM{roof-optimality}}
\label{app:roof-optimality}
  From the definition of convex roof, we have
  \begin{equation}
    \label{eq:seq_mm-3}
    \tR'(\mub) = \sup\Set{
    g(\mub):\:
    g \in \mathcal{G},\, g \leq R^\b
    }
  \enspace.
  \end{equation}
  Since we have $\tR\le R^\b$, the function $\tR$ is a valid choice for $g$ in \Eq{seq_mm-3}. This gives us $\tR\le\tR'\le R^\b$ and thus $\tR'$ must be
  consistent with $R^\b$, proving the first part.

  To prove the second part, we show a stronger statement:
  \begin{equation}
    \label{eq:seq_mm-4}
    \tD'(\mub \| \q) \leq \tD(\mub\|\q)
    \text{ for all $\mub\in\hull^\star$, $\q\in\reals^K$,}
  \end{equation}
  where $\tD'$ is the mixed Bregman divergence with respect to $\tR'$. The second part
  follows from \Eq{seq_mm-4} by setting $\q=\s$ and $\mub=\hmub^x$ (for \ZeroVal), or choosing arbitrary
  $\mub\in\hull^\star$ (for \DecVal).
  It remains to prove \Eq{seq_mm-4}.

  Since
  $\tR' \geq \tR$, we have for their conjugates $\tC'\le\tC$.
  Also, for any $\mub\in\hull^\star$ we have $\tR'(\mub) = R^\b(\mub) = \tR(\mub)$, and thus
  \begin{align}
    \tD'(\mub\|\q)
    &= \tR'(\mub) + \tC'(\q) - \q\cdot\mub
\notag
\\  &\leq \tR(\mub) + \tC(\q) - \q\cdot\mub = \tD(\mub\|\q)
\enspace.
\tag*{\qed}
  \end{align}

\subsection{PROOF OF THEOREM~\ref{thm:one-step-0-profits}}
\label{app:one-step-zero-profits}

First, assume that \CondPrice, \CondVal, and \ZeroVal are
simultaneously satisfiable using \Prot{switch}.  By
Lemma~\ref{lem:sequalsts}, this implies they are satisfiable with the
identity function for $\AdvanceState$ and some function $\AdvanceCost$.

By Lemmas~\ref{lem:one-step-condval} and~\ref{lem:roof-optimality}, it
must be the case that for any state $\s$, there exists some
$\b\in\reals^\X$, such that the conditions would remain satisfied if
$\AdvanceCost(\s)$ instead output the conjugate $\tC$ of
$\tR \coloneqq (\conv R^\b)$. \jenn{Is this obvious or should we elaborate?}  It remains to
show that the three conditions would remain satisfied if
$\AdvanceCost(\s)$ output the conjugate of $(\conv R^\hbb)$.

For all $x\in\X$, we can simplify
$\tValue(X=x;\s)$
using \Eq{D:E} and \Eq{p:E} as follows:
  \begin{align}
\tValue(X=x;\s)
  &
  = \tD(\hmub^x\|\s)
    \notag
\\&
  = \tC(\s) + \tR(\hmub^x) - \s\cdot\hmub^x
    \notag
\\&
  = \tC(\s) + R(\hmub^x) - b^x - \s\cdot\hmub^x
    \label{eq:one-step-1}
\\&
  = \tC(\s) - C(\s) + \hb^x - b^x
    \label{eq:one-step-tval-b-0}
  \end{align}
for some $\hmub^x \in \p(\Omega^x;\s)$.
  Here \Eq{one-step-1} follows by consistency of $\tR$ with $R^\b$
  and \Eq{one-step-tval-b-0} follows by the definition of $\hb^x$ in \Eq{one-step-b-0}.
  Since $\ZeroVal$ is satisfied, $\tValue(X=x;\s) = 0$ for all $x\in\X$,
  so $\tD(\hmub^x\|\s) = 0 = \tD(\hmub^{x'}\|\s)$ for all
  $x,x'\in\X$.
  \Eq{one-step-tval-b-0} then yields
  \begin{equation*}
    \label{eq:one-step-2}
    \tC(\s) - C(\s) + \hb^x - b^x = \tC(\s) - C(\s) + \hb^{x'} - b^{x'}
  \enspace.
  \end{equation*}
  Canceling the constant terms $\tC(\s)$ and $C(\s)$, we obtain that
  $\b = \hbb + c\1$ for some $c\in\reals$.
  Since $R^\b = R^{\hbb} - c$, and $\tR=(\conv R^\b)$ is consistent with $R^\b$, we conclude that $(\conv R^\hbb) = (\conv R^\b) - c$
  is consistent with $R^{\hbb}$, and therefore by
  Lemma~\ref{lem:one-step-condval}, \CondVal and \CondPrice remain
  satisfied switching to $(\conv R^{\hbb})$.
  Additionally, since $(\conv R^\b)$ and $(\conv R^{\hbb})$ differ only by
  a vertical shift, the divergences associated with both are
  identical, and \ZeroVal is also satisfied.

  For the converse, assume $\tR \coloneqq (\conv R^\hbb)$ is consistent with $R^{\hbb}$.  By
  \Lem{one-step-condval}, \CondVal and \CondPrice are satisfied,
  and it remains only to show $\tValue(X=x;\s)=0$.  This follows from
  \Eq{one-step-tval-b-0}, since now $\b=\hbb$, and by
  \Prop{one-step-roof-dual}, $\tC(\s)=C(\s)$.
  (Note that \Prop{one-step-roof-dual} is stated after
   \Thm{one-step-0-profits} in the main text,
   but its proof, given in the next section, does not rely
   on \Thm{one-step-0-profits}.)

\subsection{PROOF OF \PROP{one-step-roof-dual}}
\label{app:one-step-roof-dual}

  We first show that $(\conv R^\hbb)$ is closed.
  Since $R$ is the conjugate of $C$, it must be closed (see \App{convex}). The domain of $R$
  is $\hull$ which is polyhedral, and therefore $R$ is in
  fact continuous
  on $\hull$~(by \Prop{cont}). Since $\hull$ is compact, $R$ attains a maximum on $\hull$, and in particular
  is bounded above on $\hull$. Thus, also $R^\hbb$ is bounded above on $\hull^\star$.
  Let $u\in\reals$ be the corresponding upper bound, i.e., $R^\hbb(\mub) \leq u$ for all $\mub\in\hull^\star$.
  We may write
  $\epi\tR = \conv(\epi R^\hbb)$ by definition of the roof construction.  Now we can chop off $\epi R^\hbb$ at $u$ and consider the remainder:
\begin{align}
\notag
   S
   &=
   \set{(\mub,t) : \mub \in \hull^\star,\,R^\hbb(\mub) \leq t \leq u}
\\
\label{eq:compact:x}
   &=
   \bigcup_{x\in\X}
   \set{(\mub,t) : \mub \in \hull^x,\,R^\hbb(\x) \leq t \leq u}
\enspace.
\end{align}
  The set $S$ is compact, because it is a finite union of compact sets in \Eq{compact:x}.
  Each individual term in \Eq{compact:x} is indeed compact, because it is bounded
  (above by $u$ and below by the boundedness of $R$ on $\hull^x$) and closed (by closedness of $R$ and closedness of $\hull^x$). Since $S$ is compact, $\conv S$ is closed. Therefore,
\[
  \epi\tR=\conv(\epi R^\hbb) = (\conv S) \cup \bigParens{\hull \times [u,\infty)}
\]
  is also a closed set, and thus
  $\tR$ is a closed convex function.

Recall from standard convex analysis (see Appendix~\ref{app:convex}) that for any function $f$, we have $f^{**} = \mathrm{cl} (\conv f)$; the biconjugate of $f$ is the closed convex roof of $f$.
As we have shown, $(\conv R^\hbb)=\cl(\conv R^\hbb)=(R^{\hbb})^{**}$,
so $\tR=(R^{\hbb})^{**}$, and in particular, $\tC = \tR^* = (R^{\hbb})^{***} =
(R^{\hbb})^{*}$. Now, calculate
  \begin{align*}
    \tC(\q)
    &= \sup_{\mub\in\hull^\star} \BigBracks{\q\cdot\mub - R^{\hbb}(\mub)}\\
    &= \max_{x\in\X} \sup_{\mub\in\hull^x} \BigBracks{\q\cdot\mub - R(\mub) + \hb^x}\\
    &= \max_{x\in\X} \BigBracks{ \hb^x + \sup_{\mub\in\hull^x} \bigBracks{\q\cdot\mub - R(\mub)} }\\
    &= \max_{x\in\X} \BigBracks{ \hb^x + C^x(\q) }
  \enspace.
  \end{align*}
  Finally, observe that by definition of $\hb^x$ and \Thm{D},
\[
  \hb^x
  = C(\s) - \sup_{\mub\in\hull^x}\bigBracks{\mub\cdot\s - R(\mub)} = C(\s)-C^x(\s)
\enspace.
\]


\section{PROOFS FROM \SEC{gradual}}
\label{app:gradual}

\subsection{PROPERTIES OF LCMMS}
\label{app:lcmm}

The following properties of LCMM are used in the sequel.

\begin{theorem}
\label{thm:lcmm}
Let $C$ be a linearly constrained market maker with
\begin{equation}
\label{eq:C:lcmm}
\textstyle
  C(\q) = \inf_{\etab\in\reals^M_+}\bigBracks{\Cx(\q+\A\etab)-\b\inprod\etab}
\enspace.
\end{equation}
It
has the following properties:
\begin{enumerate}[\upshape\bfseries(a)]
\item\label{lcmm:R}
The conjugate $R$ is a restriction of $\Rx$ to $\hull$:
\[
   R(\mub)=\Rx(\mub) + \ones\Bracks{\mub\in\hull}
\enspace.
\]
\item\label{lcmm:eta:1}
For every $\q$, there exists a minimizer $\etab^\star$ of \Eq{C:lcmm}.
\item\label{lcmm:D}
Let $\etab^\star$ be a minimizer of \Eq{C:lcmm} for a specific $\q$ and
let $\deltab^\star=\A\etab^\star$. The Bregman divergence from $\q$ is then
\[
\thinskips
\textstyle
   D(\mub\|\q)\;\;\;\;=\;\;\;\;\Dx(\mub\|\q+\deltab^\star)
               + (\A^\top\mub-\b)\cdot\etab^\star
               + \ones\Bracks{\mub\in\hull}
\;\;\;\;.
\]
\item\label{lcmm:eta:2}
Let $\etab\ge\0$
and $\deltab=\A\etab$. Then $\etab$ is a minimizer of \Eq{C:lcmm}
for a specific $\q$
if and only if there exists some $\mub\in\hull$ such that
\[
\textstyle
   \Dx(\mub\|\q+\deltab)
   + (\A^\top\mub-\b)\cdot\etab
   = 0
\enspace.
\]
\end{enumerate}
\end{theorem}

Part
(\ref{lcmm:R}) shows that while the definition of $C$ in \Eq{C:lcmm}
is slightly involved, the conjugate $R$ has a natural meaning as a
restriction of the direct-sum market to the price space $\hull$. Part
(\ref{lcmm:eta:1}) shows that we can take the minimum rather than the
infimum in the definition of $C$, i.e., there is an optimal arbitrage
bundle. Part (\ref{lcmm:D}) decomposes the Bregman divergence (and
thus utility for information) into three terms. The last term forces
$\mub\in\hull$. The first term is the (direct-sum) divergence between
$\mub$ and the state resulting from the arbitrager action in the
direct-sum market. The second term is non-negative for $\mub\in\hull$,
and represents expected arbitrager gains beyond the guaranteed profit
from the arbitrage. Part (\ref{lcmm:eta:2}) spells out
first-order optimality conditions for an optimal arbitrage bundle
$\etab$.

\begin{proof} We prove the theorem in parts.

\textbf{Parts (\ref{lcmm:R}) and (\ref{lcmm:eta:1})\quad}
We use a version of Fenchel's duality from \Thm{fenchel:duality}.
Specifically, consider a fixed $\q\in\reals^K$ and
let $f$ and $g$ be defined by
\[
  f(\u)=\Cx(\q+\u)
\enspace,
\quad
  g(\etab)=\ones[\etab\ge\0]-\b\inprod\etab
\]
and hence their conjugates are
\[
  f^*(\mub)=\Rx(\mub)-\q\inprod\mub
\enspace,
\quad
  g^*(\v)=\ones[\v+\b\le\0]
\enspace.
\]
Assuming that the conditions of \Thm{fenchel:duality} are
satisfied for $f$ and $g$, and plugging in the above
definitions, we obtain
\begin{align*}
 C(\q)
 &=\inf_{\etab\in\reals^M}\BigBracks{\Cx(\q+\A\etab)-\b\inprod\etab+\ones[\etab\ge\0]}
\\
 &=\sup_{\mub\in\reals^K}\BigBracks{-\Rx(\mub)+\q\inprod\mub-\ones[\A^\top\mub-\b\ge\0]}
\\
 &=\sup_{\mub\in\reals^K}\BigBracks{\q\inprod\mub-\BigParens{\Rx(\mub)+\ones[\A^\top\mub\ge\b]}}
\enspace,
\end{align*}
showing that
\[
   \Rx(\mub)+\ones[\A^\top\mub\ge\b]
\]
is the conjugate of $C$ and the infimum in $\etab$ is attained. To finish
the proof we need to verify that the conditions of \Thm{fenchel:duality}
hold.

Note that $f$ and $g$ are closed and convex and $g$ is polyhedral. Therefore
it remains to show that there exists $\mub\in\relint(\dom f^*)$ such that
$\A^\top\mub\in\dom g^*$. Since $\dom f^*=\dom\Rx$ and
$\A^\top\mub\in\dom g^*$ if and only if $\mub\in\hull$,
it suffices to show that $\relint(\dom\Rx)\cap\hull\ne\emptyset$.

Let $\hull_g\coloneqq\set{\mub_g:\:\mub\in\hull}$
and $\hullx\coloneqq\prod_{g\in\G}\hull_g$.
By assumption, costs $C_g$ are arbitrage-free, i.e.,
$\dom R_g=\hull_g$.
For each $g$, pick $\tmub_g\in\relint\hull_g$. Since $\hull_g$ is the projection
of $\hull$ on the coordinate block $g$, there must exist $\mub^{(g)}\in\hull$
such that $\mub^{(g)}_g=\tmub_g$. Now, let
\[
   \mub^\star=\frac{1}{|\G|}\sum_{g\in\G}\mub^{(g)}
\enspace.
\]
Note that for $g'\ne g$, we have
$\mub^{(g')}_g\in\hull_g$, whereas $\mub^{(g)}_g\in\relint\hull_g$,
so $\mub^\star_g\in\relint\hull_g$ and hence
$\mub^\star\in\relint\hullx=\relint(\dom\Rx)$. At the same time
$\mub^\star\in\hull$, showing that $\relint(\dom\Rx)\cap\hull\ne\emptyset$.

\textbf{Part (\ref{lcmm:D})\quad}
Fix $\q$.
Let $\etab^\star$ be a minimizer of \Eq{C:lcmm} and
let $\deltab^\star=\A\etab^\star$. Using \Thm{lcmm}\ref{lcmm:R},
we obtain
\begin{align*}
&
  D(\mub\|\q)
\\
&\quad{}=
  R(\mub) + C(\q) - \mub\inprod\q
\\
&\quad{}=
\ones\Bracks{\mub\in\hull} + \Rx(\mub) + \Cx(\q+\deltab^\star)
\\
&\qquad\qquad{}
  - \b\inprod\etab^\star - \mub\inprod\q
\\
&\quad{}=
\ones\Bracks{\mub\in\hull} + \Rx(\mub) + \Cx(\q+\deltab^\star)
\\
&\qquad\qquad{}
  - \mub\inprod(\q+\deltab^\star)
  + \mub\inprod\deltab^\star - \b\inprod\etab^\star
\\
&\quad{}=
\ones\Bracks{\mub\in\hull} +
\Dx(\mub\|\q+\deltab^\star)
  + (\A^\top\mub-\b)\cdot\etab^\star
\enspace.
\end{align*}

\textbf{Part (\ref{lcmm:eta:2})\quad}
If $\etab$ is a minimizer of \Eq{C:lcmm} then choosing $\mub\in\nabla C(\q)$,
we have $D(\mub\|\q)=0$ and hence by \Thm{lcmm}\ref{lcmm:D}
\begin{equation}
\label{eq:app:1}
  0 = D(\mub\|\q)=\Dx(\mub\|\q+\deltab^\star)
               + (\A^\top\mub-\b)\cdot\etab^\star
\end{equation}
because, by \Thm{lcmm}\ref{lcmm:R}, $C$ is arbitrage-free, so $\mub\in\hull$.

For a converse, assume that for some $\mub\in\hull$,
$\etab\ge\0$, we have:
\begin{align}
0 &=
   \Dx(\mub\|\q+ \A\etab)
   + (\A^\top\mub-\b)\cdot\etab
\notag
\\
  &=
    \Rx(\mub) + \Cx(\q+\A\etab) - \mub\inprod(\q+\A\etab)
\notag
\\
&\qquad\qquad{}
    + (\A^\top\mub-\b)\cdot\etab
\notag
\\
  &=
    R(\mub) + \Cx(\q+\A\etab) - \mub\inprod\q
    -\b\cdot\etab
\notag
\\
  &\ge
    R(\mub) + C(\q) - \mub\inprod\q
\label{eq:app:2}
\\
  &=D(\mub\|\q)
\enspace,
\notag
\end{align}
where \Eq{app:2} is from the definition of $C$. However,
since $D(\mub\|\q)\ge 0$, we have that \Eq{app:2} holds
with the equality and hence $\etab$ is indeed the minimizer
of \Eq{C:lcmm}.
\end{proof}

\subsection{PROOF OF \THM{time-sensitive}}
\label{app:time-sensitive}

In what follows, let $C^t_g(\q_g) = \beta_g(t) C_g (\q_g /
\beta_g(t))$ and let $R^t_g$ and $D^t_g$ denote the conjugate and
divergence derived from $C^t_g$.  Define $\ttC_g$, $\ttR_g$, and
$\ttD_g$ similarly.

The definitions of $C^t$ and $\ttC$ imply that
\begin{align}
&
  \ttC_g(\q_g) = \talpha_g C^t_g(\q_g/\talpha_g)
\enspace,
\label{eq:lcmm:C}
\\
&
  \ttR_g(\mub_g) = \talpha_g R^t_g(\mub_g)
\enspace,
\label{eq:lcmm:R}
\\
&
  \ttD_g(\mub_g\|\q_g) = \talpha_g D^t_g(\mub_g\|\q_g/\talpha_g)
\enspace.
\label{eq:lcmm:D}
\end{align}
The proof proceeds in several steps:
\paragraph{Step 1}\emph{$D^t(\mub\|\s)=0$ if
and only if $\mub\in\hull$,
\begin{equation}
\label{eq:step:1}
\begin{aligned}
&
  D^t_g(\mub_g\|\s_g+\deltab^\star_g) = 0\text{ for all $g\in\G$,}
\\
&
 \text{and
 $(\A^\top\mub-\b)\cdot\etab^\star = 0$.}
\end{aligned}
\end{equation}}%

If \Eq{step:1} holds and $\mub\in\hull$, then \Thm{lcmm}\ref{lcmm:D} shows
that $D^t(\mub\|\s)=0$.
For the opposite implication note that $D^t(\mub\|\s)=\infty$ if
$\mub\not\in\hull$, so we must have $\mub\in\hull$. For $\mub\in\hull$,
by \Thm{lcmm}\ref{lcmm:D},
\begin{equation}
\label{eq:decompose:D}
D^t(\mub\|\s)
=
\sum_g D^t_g(\mub_g\|\s_g+\deltab^\star_g)
  + (\A^\top\mub-\b)\cdot\etab^\star
\enspace.
\end{equation}
Note that the last term in \Eq{decompose:D} is non-negative,
because $\etab^\star\ge 0$ and $\mub\in\hull$.
Since also the divergences $D^t_g$ are non-negative, we obtain that
all the terms must equal zero if $D^t(\mub\|\s)=0$.

\paragraph{Step 2}
$
  \etab^\star\in\argmin_{\etab\ge\0}\bigBracks{\ttCx(\tsb+\A\etab)-\b\inprod\etab}
$.

By \Thm{lcmm}\ref{lcmm:eta:2}, it suffices to exhibit $\mub\in\hull$ such that
\begin{equation}
\label{eq:step:2}
   \ttDx(\mub\|\tsb+\deltab^\star)
   + (\A^\top\mub-\b)\cdot\etab^\star
   = 0
\enspace.
\end{equation}
Pick any $\mub\in\partial C^t(\s)$, i.e., $D^t(\mub\|\s)=0$. Then by
expanding $\ttDx$ (using Eq.~\ref{eq:lcmm:D}) and then using the
definition of $\tsb$,
we obtain
\begin{align*}
&
   \ttDx(\mub\|\tsb+\deltab^\star)
   + (\A^\top\mub-\b)\cdot\etab^\star
\\
&\quad{}=
  \sum_g \talpha_g D^t_g\BigParens{\mub_g\Bigm\Vert\frac{\tsb_g+\deltab^\star_g}{\talpha_g}}
  + (\A^\top\mub-\b)\cdot\etab^\star
\\
&\quad{}=
  \sum_g \talpha_g D^t_g(\mub_g\|\s_g+\deltab^\star_g)
  + (\A^\top\mub-\b)\cdot\etab^\star .
\end{align*}
Both terms
on the right-hand side are zero by Step 1, yielding \Eq{step:2} as desired.

\paragraph{Step 3}\emph{For all $\mub\in\hull$:
\[\textstyle
 \ttD(\mub\|\tsb)
 =
 \sum_g
 \talpha_g D^t_g(\mub_g\|\s_g+\deltab^\star_g)
    + (\A^\top\mub-\b)\cdot\etab^\star
\enspace.
\]}

This follows by Step~2 and \Thm{lcmm}\ref{lcmm:D} plus Eq.~\ref{eq:lcmm:D}, noting that
\[
 (\tsb_g+\deltab_g^\star)/\talpha_g=\s_g+\deltab_g^\star
\enspace.
\]

\paragraph{Step 4}\emph{$\ttC$ and $\tsb$ satisfy \Price.}

Since $\mub\in\partial C^t(\s)$ if and only if $D^t(\mub\|\s)=0$,
and similarly for $\mub\in\partial\ttC(\tsb)$,
it suffices to show that $D^t(\mub\|\s)=0$ if and only if $\ttD(\mub\|\tsb)=0$.
First assume that $D^t(\mub\|\s)=0$. Then
Steps~1 and~3 show that $\ttD(\mub\|\tsb)=0$.
Also, vice versa: if $\ttD(\mub\|\tsb)=0$ then, from Step 3 (by a similar
reasoning as in the proof of Step 1), we have that $\mub\in\hull$,
for all $g$ it holds that $D^t_g(\mub_g\|\s_g+\deltab^\star_g)=0$, and also $(\A^\top\mub-\b)\cdot\etab^\star = 0$. Hence, by Step~1, also $D^t(\mub\|\s)=0$.

\subsection{PROOF OF \THM{gradual:partial}}
\label{app:gradual:partial}

We first show that \CondPrice and \CondVal hold. We proceed by \Prop{condval}. Fix $\x\in\X_g$, let $\Omega^\x=\set{\rhob_g=\x}$, and let $\mub\in\hull^\x\coloneqq\hull(\Omega^\x)$.
Then,
expanding $D^t(\mub\|\s)$ according to \Thm{lcmm}\ref{lcmm:D} and
$\ttD(\mub\|\tsb)$ according to \Thm{time-sensitive}, we have
\begin{align}
  D^t(\mub\|\s) - \ttD(\mub\|\tsb)
&= (1-\talpha_g)D^t_g(\mub_g\|\s_g+\deltab^\star_g)
\notag
\\
&=(1-\talpha_g)D^t_g(\x\|\s_g+\deltab^\star_g)
\label{eq:D:tD:diff}
\end{align}
which is a constant independent of the specific choice of $\mub\in\hull^\x$,
proving that both \CondPrice and \CondVal hold.

Next, we show that \DecVal holds. Let $\hmub^\x\in\p^t(\Omega^\x;\s)=\tpb(\Omega^\x;\tsb)$ (the equality holds by \CondVal).
From \Eq{D:tD:diff} and \Thm{D}, we have
\begin{align*}
&
 \ttValue(\rhob_g=\x;\tsb)
 =
 \ttD(\hmub^\x\|\tsb)\\
&\quad{}
 = D^t(\hmub^\x\|\s) - (1-\talpha_g)D^t_g(\x\|\s_g+\deltab^\star_g)
\\
&\quad{}
 = \Value(\rhob_g=\x;\s) - (1-\talpha_g)D^t_g(\x\|\s_g+\deltab^\star_g)
\enspace,
\end{align*}
i.e., the utility for event $\Omega^\x$ is non-increasing, because $\talpha_g\in(0,1)$.

In order to show \DecVal, we still
need to show that $D^t_g(\x\|\s_g+\deltab^\star_g)=0$ implies
$D^t(\hmub^\x\|\s)=0$.
Assume that $C^t_g$ is
differentiable and the submarket $g$ is tight, i.e.,
$\hull^\x=\set{\mub\in\hull:\:\mub_g=\x}$.
Assume that $D^t_g(\x\|\s_g+\deltab^\star_g)=0$.
By differentiability of $C^t_g$ this implies that
$\x=\nabla C^t_g(\s_g+\deltab^\star_g)$, and hence
$\mub_g=\x$ for all $\mub\in\partial C^t(\s)$.
By assumption, any of them is in $\hull^\x$ and hence
any of them can be chosen as a minimizer $\hmub^\x$ with
$D^t(\hmub^\x\|\s)=0$.

\subsection{BINARY-PAYOFF SUBMARKETS ARE TIGHT}
\label{app:tight}

\begin{theorem}
\label{thm:tight}
Let $g$ be a binary-payoff submarket in an LCMM, i.e.,
$\rhob_g(\omega)\in\set{0,1}^g$ for all $\omega\in\Omega$. Then $g$ is tight.
\end{theorem}
\begin{proof}
Fix any $\x\in\X_g = \set{0,1}^g$. Let $\Omega^\x\coloneqq\set{\rhob_g=\x}$, and let
$\hull^\x\coloneqq\hull(\Omega^\x)$ be the set of
beliefs consistent with $\Omega^\x$. Let $\mub\in\hull$ be such that $\mub_g=\x$.
We need to show that $\mub\in\hull^\x$.

Since $\mub\in\hull$, we can write $\mub=\sum_{\omega\in\Omega}\lambda_\omega\rhob(\omega)$
for some $\lambda_\omega\ge 0$ such that $\sum_{\omega\in\Omega}\lambda_\omega=1$.
We will argue that the condition $\mub_g=\x$ implies that $\lambda_\omega=0$ for $\omega\not\in\Omega^\x$
and thus in fact $\mub\in\hull^\x$. Essentially we show that $\hull^\x$ is the set of maximizers of a
linear function over $\hull$ and that $\mub$ is one of the maximizers.

The required linear function, $\v \cdot \mub$, is specified by the vector $\v\in\reals^K$ defined as follows:
\[
 v_i=
 \begin{cases}
 1 &\text{if $i\in g$ and $x_i=1$,}
\\
-1 &\text{if $i\in g$ and $x_i=0$,}
\\
0  &\text{if $i\not\in g$.}
\end{cases}
\]
Let $k$ be the number of $1$s in the vector $\x$. From the definition of $\Omega^\x$ we have
that $\v\inprod\rhob(\omega)=k$ for all $\omega\in\Omega^\x$. Let $\omega'\not\in\Omega^\x$,
i.e., $\omega'\in\Omega^{\x'}$ for some $\x'\in\X_g\wo\set{\x}$. Since $\x'\in\set{0,1}^g$ but
$\x'\ne\x$, there exists $i\in g$ such that $x_i=1$ but $x'_i=0$, or such that $x_i=0$ but $x'_i=1$.
Thus, $\v_g\inprod\x'\le k-1$ and hence also $\v\inprod\rhob(\omega')\le k-1$. This yields
\begin{align*}
   \v\inprod\mub
   &=
   \v\inprod
   \Parens{\sum_{\omega\in\Omega}\lambda_\omega\rhob(\omega)}
   =
   \sum_{\omega\in\Omega}\lambda_\omega\BigParens{\v\inprod\rhob(\omega)}
\\
   &\le
   k\Parens{\sum_{\omega\in\Omega^\x}\lambda_\omega}
   +(k-1)\Parens{\sum_{\omega'\in\Omega\wo\Omega^\x}\lambda_{\omega'}}
\\
   &=
   k - \sum_{\omega'\in\Omega\wo\Omega^\x}\lambda_{\omega'}
\enspace.
\end{align*}
However, $\mub_g=\x$ and thus $\v\inprod\mub=k$. Therefore, we must have
$\lambda_{\omega'}=0$ for $\omega'\in\Omega\wo\Omega^\x$, proving the theorem.
\end{proof}


\section{CONDITIONAL PRICE VECTORS}
\label{app:cond:price}

\subsection{CONDITIONAL PRICES FOR LMSR}

In this section we show that conditional price
vectors for LMSR coincide with conditional probabilities.

Recall that for LMSR, we have
the outcomes $\Omega=[K]$, payoffs
$\rho_i(\omega)=\1[\omega=i]$, and prices
\[
  p_i(\q)=\frac{e^{q_i}}{\sum_{j\in[K]} e^{q_j}}
\enspace,
\]
i.e., price vectors are probability distributions
over $i\in[K]$.

The mixed Bregman divergence for LMSR has
the form
\[
  D(\mub\|\q)
  =\sum_{i\in[K]} \mu_i\ln\Parens{\frac{\mu_i}{p_i(\q)}}
  =
  \KL\bigParens{\mub\bigm\|\p(\q)}
\]
where $\KL(\mub\|\nub)\coloneqq\sum_{i\in[K]}\mu_i\ln(\mu_i/\nu_i)$
is the KL divergence defined for any pair of distributions $\mub$,
$\nub$ on $[K]$.
KL divergence is always non-negative, possibly equal to $\infty$,
and equal to zero if and only if $\mub=\nub$.

Let $\q\in\reals^K$ and $\event\subseteq[K]=\Omega$ be a non-null event.
Let $\hmub$ be the probability vector obtained by conditioning $\p(\q)$
on the event $\event$, i.e.,
\[
  \hmu_i=
\begin{cases}
    p_i(\q) / c
    &
    \text{if $i\in\event$,}
\\
    0
    &
    \text{otherwise,}
\end{cases}
\]
where $c=\sum_{i\in\event} p_i(\q)$ is the normalization over $\event$.
Note that $p_i(\q)>0$ for all $i\in[K]$, so $c>0$.
We will now argue that $\p(\event;\q)=\set{\hmub}$.

We appeal to \Eq{p:E} of \Thm{D}. Specifically, we will show that $\hmub$ is the
unique minimizer of $\min_{\mub'\in\hull(\event)} D(\mub'\|\q)$.

First, note that $\hmub\in\hull(\event)$, and from the definition of $\hmub$
\begin{align*}
D(\hmub\|\q)
  &=
  \sum_{i\in\event}
  \hmu_i\ln\Parens{\frac{\hmu_i}{p_i(\q)}}
  =
  \sum_{i\in\event}
  \hmu_i\ln\Parens{\frac{p_i(\q)/c}{p_i(\q)}}
\\
  &=
  \sum_{i\in\event}
  \hmu_i\ln(1/c)
  =
  \ln(1/c)
\enspace.
\end{align*}
Now, let $\mub'\in\hull(\event)$ and
compare the values $D(\mub'\|\q)$ and $D(\hmub\|\q)$:
\begin{align*}
&
  D(\mub'\|\q)-D(\hmub\|\q)
  =
  \Parens{\sum_{i\in\event}
  \mu'_i\ln\Parens{\frac{\mu'_i}{p_i(\q)}}}
  -
  \ln(1/c)
\\
&\quad{}=
  \Parens{\sum_{i\in\event}
  \mu'_i\ln\Parens{\frac{\mu'_i}{p_i(\q)}}}
  -
  \Parens{\sum_{i\in\event}
  \mu'_i\ln(1/c)}
\\
&\quad{}=
  \sum_{i\in\event}
  \mu'_i\ln\Parens{\frac{\mu'_i}{p_i(\q)/c}}
\\
&\quad{}=
  \sum_{i\in\event}
  \mu'_i\ln\Parens{\frac{\mu'_i}{\hmu_i}}
 =\KL(\mub'\|\hmub)
\enspace.
\end{align*}
Thus, we have $D(\mub'\|\q)\ge D(\hmub\|\q)$ with
equality if and only if $\mub'=\hmub$, i.e.,
$\hmub$ is the sole minimizer of $\min_{\mub'\in\hull(\event)} D(\mub'\|\q)$.

\subsection{OPTIMAL TRADING GIVEN \texorpdfstring{$\event$}{E}}

In this section we analyze the prices that result from actions of a
trader optimizing his guaranteed profit from the information
$\omega\in\event$ as in \Def{event-value} in the market with cost
function $C$.   Intuitively, we would like to say that such a trader would
move the market price to a conditional price vector $\hmub \in
\p(\event;\q)$.  However, this may not be possible.  For example,
consider a complete market using LMSR.  In such a market, a trader
can push the market price arbitrarily close to any $\mub \in
\hull(\event)$, but cannot push the price all the way to $\mub$ with
any finite purchase (unless $\event = \Omega$).

Because of this,
instead of reasoning directly about finite purchases, we introduce the
notion of an \emph{optimizing action sequence} and show that in the
limit such a trader would move the market
from a state $\q$ to
states that minimize the Bregman
divergence to conditional price vectors
$\hmub\in\p(\event;\q)$. Then we
argue that for $R$ strictly convex
(such as entropy in case of LMSR),
this implies that the resulting market price vector
approaches the unique conditional price vector in the limit.

We begin by formalizing the optimizing behavior in
\Def{event-value}.
\begin{definition}
\label{def:opt:seq}
We say that $\set{\r_i}_{i=1}^\infty$ is an
\emph{optimizing action sequence}
with respect to a non-null event $\event$ and a state $\q$ if
\[
  \adjustlimits\lim_{i\to\infty}\min_{\omega\in\event}
   \BigBracks{
     \rhob(\omega)\inprod\r_i - C(\q+\r_i) + C(\q)}
   =
   \Value(\event;\q)
\enspace.
\]
We say that $\set{\q_i}_{i=1}^\infty$
is an \emph{optimizing state sequence} with respect to $\event$
and $\q$ if
\[
  \adjustlimits\lim_{i\to\infty}\min_{\omega\in\event}
   \BigBracks{
     \rhob(\omega)\inprod(\q_i-\q) - C(\q_i) + C(\q)}
   =
   \Value(\event;\q)
\enspace.
\]
(Thus, any optimizing action sequence yields an
optimizing state sequence $\q_i=\q+\r_i$
and vice versa.)
\end{definition}

We next show that optimizing state sequences minimize
divergence to conditional price vectors. Specifically,
the divergence between any state sequence and any
conditional price vector tends to zero. Loosely
speaking, this means that the market is moving towards
states whose associated prices, in the limit, include
all conditional price vectors.

\begin{theorem}
\label{thm:opt:conv:D}
Let $\set{\q_i}_{i=1}^\infty$
be an optimizing state sequence with respect to $\event$
and $\q$, and let $\hmub\in\p(\event;\q)$. Then
$D(\hmub\|\q_i)\to 0$ as $i\to\infty$.
\end{theorem}

\begin{proof}
Since the minimized objective in \Def{opt:seq} is linear
in $\rhob(\omega)$, we can without loss of generality
replace minimization over $\rhob(\omega)$ where $\omega\in\event$
by minimization over $\mub'\in\hull(\event)$, and thus assume
that
\begin{equation}
\label{eq:opt:1}
\begin{aligned}
&
\adjustlimits\lim_{i\to\infty}\min_{\mub'\in\hull(\event)}
   \BigBracks{
     \mub'\inprod(\q_i-\q) - C(\q_i) + C(\q)}
\\
&\hspace{0.6\columnwidth}{}=
   \Value(\event;\q)
\,.
\end{aligned}
\end{equation}
The expression in the brackets can be rewritten as
\[
     \mub'\inprod(\q_i-\q) - C(\q_i) + C(\q)
     =
     -D(\mub'\|\q_i) + D(\mub'\|\q)
\,.
\]
Furthermore, by \Thm{D}, we have $\Value(\event;\q)=D(\hmub\|\q)$.
We can therefore rewrite \Eq{opt:1} as
\begin{equation}
\label{eq:opt:2}
\adjustlimits\lim_{i\to\infty}\min_{\mub'\in\hull(\event)}
   \BigBracks{
     D(\mub'\|\q) - D(\mub'\|\q_i)
   }
   =
   D(\hmub\|\q)
\,.
\end{equation}
To get the statement of the theorem, note that for all $i$,
\begin{align}
\label{eq:opt:3}
  D(\hmub\|\q)
  &\ge
  D(\hmub\|\q) - D(\hmub\|\q_i)
\\
\label{eq:opt:4}
  &\ge
  \min_{\mub'\in\hull(\event)}
  \BigBracks{
     D(\mub'\|\q) - D(\mub'\|\q_i)
  }
\end{align}
where \Eq{opt:3} follows by non-negativity of
the divergence, and \Eq{opt:4} because $\hmub\in\hull(\event)$.
Since \Eq{opt:4} converges to $D(\hmub\|\q)$ by \Eq{opt:2},
we obtain that the right hand-side in \Eq{opt:3} must also
converge to $D(\hmub\|\q)$,
i.e., $D(\hmub\|\q_i)\to 0$.
\end{proof}

When $R$ is strictly convex on $\hull$, \Thm{opt:conv:D} can be
strengthened to show that $\p(\q_i)\to\hmub$. Strict convexity of $R$
is equivalent to a certain notion of smoothness of $C$.  It is
stronger than differentiability of $C$~\citep[][Theorem 26.3]{Rockafellar70},
but weaker than the existence of a Lipschitz-continuous gradient for
$C$.~\footnote{%
Proposition 12.60ab,
R.~Tyrrell Rockafellar, Roger J.-B. Wets.
\emph{Variational analysis}.
Springer, 1998.}

\begin{theorem}
\label{thm:opt:conv:p}
Let $\set{\q_i}_{i=1}^\infty$
be an optimizing state sequence with respect to $\event$
and $\q$, and let $\hmub\in\p(\event;\q)$. If $R$
is strictly convex on $\hull$ then
$\p(\q_i)\to\hmub$ as $i\to\infty$.
\end{theorem}

\begin{proof}
First note that if $R$ is strictly convex then
$C$ is differentiable~\citep[][Theorem 26.3]{Rockafellar70}, and thus $\p(\q_i)$
is always a singleton.
Next note that the sequence
$\set{\p(\q_i)}_{i=1}^\infty$ is
contained in a compact set $\hull$, so it
must have a cluster point in $\hull$. Pick
an arbitrary cluster point $\mub^\star$ and choose a
subsequence $\set{\q_{i(j)}}_{j=1}^\infty$ such that $\p(\q_{i(j)})\to\mub^\star$
as $j\to\infty$.
We will show that $\mub^\star=\hmub$ and thus all of the cluster
points of the original price sequence $\set{\p(\q_i)}_{i=1}^\infty$
coincide. This implies that the sequence
actually converges to $\hmub$ (again, because it is contained
in a compact set $\hull$).

To simplify writing, let $\q'_j\coloneqq\q_{i(j)}$
and $\mub'_j\coloneqq\p(\q'_j)$. By the choice of the subsequence,
we have $\mub'_j\to\mub^\star$. By \Prop{first:order}, we have $\q'_j\in\partial R(\mub'_j)$
and by convexity of $R$ we have the lower bound
\begin{equation}
\label{eq:opt:lower}
 R(\mub)\ge R(\mub'_j)+(\mub-\mub'_j)\inprod\q'_j
\end{equation}
valid for all $\mub$. We will analyze the limits of this lower bound on the line segment
connecting $\mub^\star$ and $\hmub$ to argue that
$R$ must be linear on this line segment. This will yield
a contradiction unless $\mub^\star=\hmub$.

By \Thm{opt:conv:D} we have that
$D(\hmub\|\q_i)\to 0$ and hence also $D(\hmub\|\q'_j)\to 0$.
To begin the analysis of the lower bound in \Eq{opt:lower},
we rewrite $D(\hmub\|\q'_j)$ as
\begin{align}
\notag
  D(\hmub\|\q'_j)
  &=
  R(\hmub)+C(\q'_j)-\hmub\inprod\q'_j
\\
\label{eq:opt:5}
  &=
  R(\hmub)-R(\mub'_j)+(\mub'_j-\hmub)\inprod\q'_j
\end{align}
where the last equality follows because
$C(\q'_j)=\mub'_j\inprod\q'_j-R(\mub'_j)$ by \Prop{first:order}.
Since $D(\hmub\|\q'_j)\to 0$, \Eq{opt:5} yields
\begin{equation}
\label{eq:opt:6}
  \lim_{j\to\infty}\Bracks{R(\mub'_j)+(\hmub-\mub'_j)\inprod\q'_j}
  =
  R(\hmub)
\enspace.
\end{equation}
Thus, we see that the lower bound of \Eq{opt:lower} at $\mub=\hmub$
is tight as $j\to\infty$.

Next, we note that $R$ is continuous on $\hull$ by \Prop{cont},
because $\hull$ is polyhedral.

We now focus on the line segment
connecting $\mub^\star$ and $\hmub$. Let $\lambda\in[0,1]$ and consider
\Eq{opt:lower} at $\mub'_j(\lambda)\coloneqq (1-\lambda)\mub'_j + \lambda\hmub$:
\begin{align*}
R(\mub'_j(\lambda))
  &\ge
  R(\mub'_j) + (\mub'_j(\lambda)-\mub'_j)\inprod\q'_j
\\
  &=
  R(\mub'_j) + \lambda(\hmub-\mub'_j)\inprod\q'_j
\\
  &=
  (1-\lambda)R(\mub'_j)
\\
  &\qquad{}
  + \lambda\BigParens{R(\mub'_j) + (\hmub-\mub'_j)\inprod\q'_j}
\enspace,
\end{align*}
where the first equality follows from the definition of $\mub'_j(\lambda)$ and
the second by rearranging the terms. Taking $j\to\infty$ and using
\Eq{opt:6} and the continuity of $R$, we obtain
\[
   R\BigParens{(1-\lambda)\mub^\star+\lambda\hmub}
   \ge
   (1-\lambda)R(\mub^\star)
   +\lambda R(\hmub)
\enspace.
\]
However, by convexity we also have
\[
   R\BigParens{(1-\lambda)\mub^\star+\lambda\hmub}
   \le
   (1-\lambda)R(\mub^\star)
   +\lambda R(\hmub)
\enspace,
\]
so indeed $R$ must be linear on the line segment
connecting $\mub^\star$ and $\hmub$, which contradicts
strict convexity of $R$ unless $\mub^\star=\hmub$.
\end{proof}

\section{ROBUST BAYES UTILITY}
\label{app:value}

\ignore{
In \Sec{formalism}, we motivate the definitions of $\Value(\event;\q)$,
$\Value(\mub;\q)$, and $\Value(\mub\given\event;\q)$ from the
perspective of the trader, but it is possible to give a symmetric
treatment from the perspective of the market maker. For instance,
given that the market maker wishes to gather information, and assuming
that the consensus price in the market (following a sequence of trades)
correctly estimates the expected price, the market maker's expected
loss will exactly coincide with the value of belief as defined above.
As the market maker is free to choose $C$, we can conclude that these
quantities capture how much the market maker should value the
consensus price.
} 

In \Sec{formalism}, we motivate the utility for information
as the market maker's willingness to pay for information,
or, equivalently, as the traders' ability to profit
from their information. Another motivation for the same
definitions, pursued in \Sec{bregman}, arises from defining
the utility for information via a
measure of distance, such that the market maker is
willing to pay more for the information more distant
from the current state.

In this section, we give a fourth motivation, showing how our
definitions naturally match up with concepts from robust Bayes
decision theory~\cite{grunwald2004game}.  In \Sec{formalism}, we
adopted the perspective of either an expected-utility-maximizing
trader (for the utility of a belief) or a worst-case trader (for the
utility of an event).  Here we show that if we make a slightly
stronger assumption about the behavior of traders endowed with various
information relevant to the market maker, these two notions can be
unified. Specifically, we will show that assuming that the traders are
robust Bayes decision makers, we obtain the same definitions of the
utility for information.

As before, let $\Omega$ be a finite set of outcomes.
Let $\Delta$ be the set of probability distributions
over $\Omega$. Consider a decision maker trying to
choose an action $a$ from some action set before an
outcome is realized. Given an action $a$ and a
realized outcome $\omega\in\Omega$, the decision maker
receives the utility $u(a,\omega)$. We assume that the
decision maker's information $\I$ is represented as
a non-null subset of $\Delta$, i.e.,
$\emptyset\ne\I\subseteq\Delta$.
The decision maker assumes that the outcome $\omega$
is drawn according to some probability distribution $P$,
but the only information about $P$ is that $P\in\I$. Given
this information, we call the decision maker the
\emph{robust Bayes decision maker} if he is trying
to maximize the worst-case expected utility where the worst case
is over $P\in\I$. The obtained
worst-case expected utility is referred to as the
\emph{robust Bayes utility for $\I$} and defined
as
\[
  \RBV(\I)
  \coloneqq
  \adjustlimits\sup_a\inf_{P\in\I}\E_{\omega\sim P}[u(a,\omega)]
\enspace.
\]

Consider a prediction market with the cost function $C$
and the current state $\q$.
Actions available to a trader
are all possible trades $\r\in\reals^K$, and the utility of the trader is
\[
  u(\r,\omega) = \rhob(\omega)\inprod\r - C(\q+\r) + C(\q)
\enspace.
\]
To see that our utility for information is actually
the robust Bayes utility, define the following information sets:
\begin{gather*}
  \set{\E_P[\rhob]=\mub}
  \coloneqq
  \set{P\in\Delta:\:\E_P[\rhob]=\mub}
\\
  \set{P[\event]=1}
  \coloneqq
  \set{P\in\Delta:\:P[\event]=1}
\enspace.
\end{gather*}
The first corresponds to the probability distributions $P$
that give rise to the expected value $\E_P[\rhob]=\mub$;
the second corresponds to the probability distributions
that put all of their mass on outcomes $\omega\in\event$.
Plugging these information sets into the definition
of the robust Bayes utility, we obtain
\begin{gather*}
  \RBV(\E_P[\rhob]=\mub)=\Value(\mub;\q)
\\
  \RBV(P[\event]=1)=\Value(\event;\q)
\enspace.
\end{gather*}
Thus indeed the market maker's utility for
a belief and for an event is a robust Bayes utility.

While the notion of excess utility is not standard in robust Bayes
decision theory, it can be naturally defined
as follows. Let $\I_1,\I_2\subseteq\Omega$
such that $\I_1\cap\I_2\ne\emptyset$. Then
the \emph{excess robust Bayes utility for $\I_1$ given $\I_2$}
is
\[
  \RBV(\I_1\given\I_2)
  =
  \RBV(\I_1\cap\I_2)
  -
  \RBV(\I_2)
\enspace,
\]
and thus we also obtain
\[
  \RBV\bigl(\E_P[\rhob]=\mub \,\bigm|\, P[\event]=1 \bigr)
  =
  \Value(\mub\given\event;\;\q)
\enspace.
\]

\citet{grunwald2004game} show that whenever the set $\I$ is
closed and convex, the robust Bayes utility $\RBV(\I)$
coincides with the dual
concept of the \emph{maximum (generalized) entropy}, which
seeks to find the distribution of the maximum entropy
that satisfies a given set of constraints (expressed
as $\I$).
We do not go into details here, but simply point out
that the correspondence between the utility of information
and the Bregman divergence (\Thm{D}) is just
a special case of the duality between the robust Bayes
and the maximum entropy.


%
%


\section{SUFFICIENT CONDITIONS AND ROOF EXAMPLES}
\label{app:ex-roof}

Here we explore when we can and cannot achieve implicit
submarket closing, i.e., \ZeroVal, \CondVal, and
\CondPrice simultaneously, in the sudden revelation setting.  We begin
with an example in which implicit submarket closing is not possible,
and then present
sufficient conditions, followed by additional examples.

\subsection{IMPOSSIBILITY EXAMPLE}
\label{app:closing:impossible}

\begin{example}
  \label{ex:one-step-count}
  \raf{IF TIME/SPACE: add figure} 
  Consider the square market introduced in \Ex{formalism-square}
  with the observation function
  $X(\omegab) = \omega_1 + \omega_2 \in \{0,1,2\}$.  We will see
  that for this market,
the condition of Theorem~\ref{thm:one-step-0-profits}
  cannot be satisfied and therefore
  we cannot achieve \CondVal.
  Specifically, we show that there
  exists an $\s$ for which no convex function is
  consistent with $R^\hbb$.

  First note that the observation function gives rise to conditional price spaces $\hull^0 =
  \set{(0,0)}$, $\hull^1 = \conv\set{(1,0),\,(0,1)}=\{(\lambda,1-\lambda) : \lambda\in[0,1]\}$, and
  $\hull^2 = \{(1,1)\}$.
\ignore{
The resulting submarket cost functions $C^x$ are then
\begin{align*}
  C^0(\q) &= 0\enspace,
\\C^1(\q) &= 2\ln(e^{q_1/2}+e^{q_2/2})\enspace,
\\C^2(\q) &=q_1+q_2\enspace.
\end{align*}
}
We examine the value of $R^\hbb$ at three points,
\[
  \mub^0=(0,0)
  \enspace,\quad
  \mub^1=(\tfrac 1 2, \tfrac 1 2)
  \enspace,\quad
  \mub^2=(1,1)
  \enspace.
\]
  By
  Proposition~\ref{prop:one-step-roof-dual}, we have
  $\hb^x=C(\s)-C^x(\s)$, and so
\begin{align*}
  R^\hbb(\mub^{0})
&=
  R(0,0) - [C(\s)-C^0(\s)] = -C(\s),
\\
  R^\hbb(\mub^{2})
  &=
  R(1,1) - [C(\s)-C^2(\s)]
= -C(\s) + s_1 + s_2,
\\
  R^\hbb(\mub^{1})
&=
  R(\tfrac12,\tfrac12) - [C(\s)-C^1(\s)]
\\&=
  -2\ln 2 - C(\s) + 2\ln\bigParens{e^{s_1/2}+e^{s_2/2}}
\\&=
  -2\ln 2 - C(\s)
\\&\qquad{} + 2\ln\Bracks{e^{(s_1+s_2)/4} \Parens{ e^{(s_1-s_2)/4}+e^{(s_2-s_1)/4} }}
\\&=
  -C(\s) + \tfrac{s_1+s_2}{2} +
  2\ln\bigParens{\tfrac{z+z^{-1}}{2}},
\end{align*}
  where $z=e^{(s_1-s_2)/4}$. Note that $\mub^1=(\mub^0+\mub^2)/2$,
  but $R^\hbb(\mub^1)>(R^\hbb(\mub^0)+R^\hbb(\mub^2))/2$
  whenever $z+z^{-1}>2$, i.e., whenever $z>0$ and $z\ne 1$.
  From the definition of $z$ this happens whenever $s_1\ne s_2$,
  so for any such~$\s$, no convex function can
  be consistent with $R^\hbb$.
\end{example}

\subsection{SUFFICIENT CONDITIONS}
\label{app:extend-R-b}

As we saw in Example~\ref{ex:one-step-count}, there is sometimes
tension between satisfying \ZeroVal and \CondVal, and in particular, we
cannot always achieve both.  We now establish \emph{sufficient}
conditions under which we can achieve both of these goals (and hence
\CondPrice as well).  We will do this in a way that focuses on the
geometry of the sets $\hull^x$, and consequently our results will
apply regardless of the choice of $C$ and the transition state $\s$.
This not only simplifies the theory, but has practical advantages as
well; the market designer need not worry about the transition
state, and can choose $C$ independently of concerns about implicit
market closing.

In particular, we will show sufficient conditions for when $(\conv R^\hbb)$
is consistent with $R^{\hbb}$, and then apply Theorem~\ref{thm:one-step-0-profits}.
In fact, we show something stronger, by characterizing when $(\conv R^\b)$ is
consistent with $R^\b$ for \emph{all} vectors~$\b$.

Recall that a \emph{face} of a convex set $S$ is a convex subset $F\subseteq S$
such that any line segment in $S$ whose relative interior intersects $F$,
must have both of its endpoints in $F$. Our sufficient condition
requires that the sets $\hull^x$ be faces of $\hull$. This means that
elements of
$\hull^x$ cannot be obtained as convex combinations including
elements from $\hull^y$ for $y\ne x$ with non-zero weight.

We define simplices $\Delta_\X \coloneqq \{\lub\in\reals_+^\X:\: \sum_x \lambda_x \leq 1\}$
and $\Delta_k \coloneqq \{\lub\in\reals_+^k:\: \sum_i \lambda_i \leq 1\}$ where $\reals_+$ are non-negative reals.
Before
proving the sufficient condition, we state the following alternative characterization of
the face.

\begin{proposition}
\label{prop:face}
Let $F$ and $S$ be convex sets and $F\subseteq S$. Then $F$ is a face of $S$
if and only if for all $\mub\in F$, any decomposition of $\mub$ into a convex combination
over $S$ must put zero weight on points outside $F$; i.e., for
all $k\ge 1$, $\lub\in\Delta_k$ and $\mub^i\in S$
such that $\mub=\sum_{i=1}^k \lambda_i\mub^i$, we
must have that $\lambda_i=0$ for $\mub^i\not\in F$.
\end{proposition}
\begin{proof}
Assume first that $F$ is a face. By convexity of $F$, a convex combination
of any points from $F$ lies in $F$. Also, any convex combination of points
from $S\wo F$ must lie in $S\wo F$. This is true
for $k=2$ points by the definition of the face.
For $k>2$ it follows by induction,
because, assuming $\lambda_1>0$, we can rewrite the convex combination
of $\mub_i\in S\wo F$ as
\begin{multline*}
  \lambda_1\mub_1+\dotsb+\lambda_k\mub_k
\\
{}=
  (1-\lambda_k)\Bracks{\frac{\lambda_1\mub_1+\dotsb+\lambda_{k-1}\mub_{k-1}}{\lambda_1+\dotsb+\lambda_{k-1}}}
  +\lambda_k\mub_k
\enspace.
\end{multline*}
The term in the brackets is in $S\wo F$ by the inductive hypothesis, so
the entire expression is a convex combination of $k=2$ points from $S\wo F$,
and therefore lies in $S\wo F$ by the definition of the face.
Now assume that $\mub\in F$, and consider any decomposition of $\mub$
into a convex combination over $S$. By the above reasoning,
we can collect the terms with $\mub^i\in F$ and $\mub^i\not\in F$ and
write $\mub=\lambda_F\mub^F+\lambda_{S\wo F}\mub^{S\wo F}$ where
$\lambda_F$ and $\lambda_{S\wo F}$ are the respective
sums of weights of $\mub^i\in F$ and $\mub^i\in S\wo F$, and
$\mub^F\in F$ and $\mub^{S\wo F}\in S\wo F$ are their respective
convex combinations. From the definition
of the face, we obtain $\lambda_{S\wo F}=0$.

For the opposite direction, consider any $\mub^1,\mub^2\in S$ and assume
that a point $\mub$ in the relative interior of the connecting line segment
lies in $F$, i.e., $\mub=\lambda_1\mub^1+\lambda_2\mub^2$ with $\lambda_1,\lambda_2>0$.
The condition of the proposition then implies that the endpoints
$\mub^1,\mub^2$ be in $F$, so
$F$ must be a face.
\end{proof}

\begin{proposition}
  \label{prop:extend-R-b}
  For any convex $R$ with $\dom R=\hull$, $(\conv R^\b)$ is consistent with $R^\b$ for all $\b\in\reals^\X$ if and only if
  the sets $\hull^x$ are disjoint faces of $\hull$.~\footnote{%
If $R^\b$ is not well defined, we assume that no function can be consistent with $R^\b$.%
}
\end{proposition}
\begin{proof}
    Suppose that the sets $\hull^x$ are disjoint faces of $\hull$, and $R$ and $\b$ are given. By
    Proposition B.2.5.1 of~\citet{urruty2001fundamentals}, we may use an alternate representation of the convex roof,
  \begin{align*}
    (\conv R^\b)(\mub)
    &
      = \inf\left\{\sum_{i=1}^k \!\lambda_i R^\b(\mub^i):\:
           k\!\ge\!1,\, \mub^i\!\!\in\!\hull^\star,\right.
\\[-4pt]
    &\qquad\quad
    \left. \lub\!\in\!\Delta_k,\,
    \sum_{i=1}^k \!\lambda_i \mub^i\!\!=\! \mub \right\}
.
  \end{align*}
  Intuitively, this expression examines all upper bounds imposed by the convexity constraints from $R^\b$ and defines $(\conv R^\b)$ as the infimum of these upper bounds.  Note that $R^\b$ is convex on each of the sets $\hull^x$ (since it is just a shifted copy of $R$ on $\hull^x$). Therefore, we may condense convex combinations within each $\hull^x$ (which only lowers the corresponding $R^\b$ values), yielding
  \begin{align}
\notag
    (\conv R^\b)(\mub)
    &
      = \inf\left\{\sum_{x\in\X} \!\lambda_x R^\b(\mub^x):\:
      \mub^x\!\!\in\!\hull^x,\right.
    \\[-4pt]
\label{eq:one-step-5}
    &\qquad\quad
    \left. \lub\!\in\!\Delta_\X,\,
    \sum_{x\in\X} \!\lambda_x\mub^x\!\!=\! \mub \right\}
.
   \end{align}
  For a given $y\in\X$, the set $\hull^y$ is a face disjoint from
  all $\hull^x$ for $x\ne y$. Thus, if $\mub\in\hull^y$,
  we obtain by \Prop{face} that the $\lub$ in the
  right hand side of \Eq{one-step-5} must have $\lambda_x=0$
  for $x\ne y$ and $\lambda_y=1$. This immediately yields
  $(\conv R^\b)(\mub) = R^\b(\mub)$.

  For the other direction, first note that if sets $\hull^x$
  are not disjoint then $R^\b$ is not well defined for all
  $\b$ and the theorem holds. Assume
  that sets $\hull^x$ are disjoint, but they are not all faces.
  Therefore, for some $y\in\X$, we have $\mub\in\hull^y$ which
  can be written as a convex combination
  $\mub=\lambda_1\mub^1+\lambda_2\mub^2$
  with $\lambda_1,\lambda_2>0$, $\mub^1, \mub^2 \in \hull$, but $\mub^1\not\in\hull^y$.
  We will argue that this implies that
  $\mub$ can be written as a convex combination
  across $\mub^x\in\hull^x$, putting non-zero weight on some $\mub^z$
  where $z\ne y$. The reasoning is as follows.
  Since $\mub^1,\mub^2\in\hull$, they can be written as convex
  combinations of $\rhob(\omega)$ across $\omega\in\Omega$.
  Collecting $\omega\in\Omega^x$ for $x\in\X$, vectors
  $\mub^1$ and $\mub^2$ can be in fact
  written as convex combinations
\[
  \mub^1 = \sum_{x\in\X} \lambda_{1,x}\mub^{1,x}
\enspace,
\quad
  \mub^2 = \sum_{x\in\X} \lambda_{2,x}\mub^{2,x}
\]
where $\mub^{1,x},\mub^{2,x}\in\hull^x$.
Collecting the matching terms, we can thus write $\mub$ as
\[
  \mub = \sum_{x\in\X} \lambda_x\mub^x
\]
where $\lambda_x=\lambda_1\lambda_{1,x}+\lambda_2\lambda_{2,x}$
and
\[
  \mub^x =
  \frac{\lambda_1\lambda_{1,x}\mub^{1,x} + \lambda_2\lambda_{2,x}\mub^{2,x}}
       {\lambda_1\lambda_{1,x}+\lambda_2\lambda_{2,x}}
  \in \hull^x
\enspace.
\]
Since $\mub^1\not\in\hull^y$, we must have $\lambda_{1,y}<1$, and thus also
$\lambda_y<1$ (because $\lambda_1>0$). Hence, there must exist some $z\ne y$ such that
$\lambda_z>0$.

To show that $(\conv R^\b)$ cannot be consistent with $R^\b$ for all~$\b$,
consider $\b$ with $b^x = 0$ for $x\neq z$ and $b^z$ equal to some large value. Thus, $\sum_x \lambda_x R^\b(\mub^x) = \sum_x \lambda_x R(\mub^x) - \lambda_z b^z$.  We may make this expression as low as desired by increasing $b^z$, and in particular, for a sufficiently large~$b^z$, we have $\sum_x \lambda_x R^\b(\mub^x) < R(\mub) = R^\b(\mub)$, so any function which is consistent with $R^\b$ will not be convex.
\end{proof}

Combining Theorem~\ref{thm:one-step-0-profits} and Proposition~\ref{prop:extend-R-b},
we have the following theorem.

\begin{theorem}
  \label{thm:one-step-all-goals}
  If the sets $\hull^x$ are disjoint faces of $\hull$,
  then \CondPrice, \CondVal, and \ZeroVal are achieved with
  $\AdvanceState$ as the identity and
  $\AdvanceCost$ outputting the conjugate
of $\tR = (\conv R^{\hbb})$.
\end{theorem}

\subsection{BINARY-PAYOFF LCMMS AND THE SIMPLEX}
\label{app:suff:examples}

Two key examples studied in this paper are the LMSR on the simplex and
LCMMs. In this section, we show that the sufficient condition
introduced in the previous section holds for LCMMs with binary payoffs
when the payoffs of one submarket are observed, as well as for any
observations on a simplex.

We will argue by \Thm{one-step-all-goals}, showing that the sets
$\hull^x$ are exposed faces of $\hull$. Recall that $F$ is an
\emph{exposed face} of a convex set $S$ if $F$ is the set of
maximizers of some linear function over $S$. The exposed face is
always a face~\cite[][page 162]{Rockafellar70}

Instead of working with $\hull^x$, it in fact suffices
to work with $\Omega^x$. Inspired by the definition of an exposed
face, we define an ``exposed event'' as follows.
\begin{definition}
\label{def:exposed}
An event $\event\subseteq\Omega$ is called \emph{exposed} if it is
the set of maximizers of some linear function of $\rhob(\omega)$, i.e.,
if there exists a vector $\v\in\reals^K$ such that
\[
  \event=\argmax_{\omega\in\Omega} \Bracks{\v\inprod\rhob(\omega)}
\enspace.
\]
\end{definition}
It is immediate that if $\Omega^x$ is an exposed event,
then $\hull^x$ is an exposed face disjoint from
$\hull^y$ for any $y\ne x$. Combining this with
\Thm{one-step-all-goals} yields the following theorem.
\begin{theorem}
\label{thm:exposed}
If all events $\Omega^x$ are exposed,
  then \CondPrice, \CondVal, and \ZeroVal are achieved with
  $\AdvanceState$ as the identity and
  $\AdvanceCost$ outputting the conjugate
of $\tR = (\conv R^{\hbb})$.
\end{theorem}

We next show how \Thm{exposed} can be used to argue
that submarket closing is possible in binary-payoff LCMMs
and on a simplex.

\begin{example}\emph{Submarket closing in binary-payoff LCMMs.}
We need to argue that the events corresponding
to submarket observations in a binary-payoff
($\rhob(\omega)\in\set{0,1}^K$ for all $\omega\in\Omega$) LCMMs  are exposed. We use the same construction
as in the proof of \Thm{tight}.
Let $g$ be a submarket in a binary-payoff LCMM.
Let $\x\in\X_g$ and $\Omega^\x\coloneqq\set{\rhob_g=\x}$.
We need to show that $\Omega^\x$ is exposed.
Consider $\v\in\reals^K$ with the components
\[
 v_i=
 \begin{cases}
 1 &\text{if $i\in g$ and $x_i=1$,}
\\
-1 &\text{if $i\in g$ and $x_i=0$,}
\\
0  &\text{if $i\not\in g$.}
\end{cases}
\]
Let $k$ be the number of $1$s in $\x$.
Now, as in the proof of \Thm{tight},
we have $\v\inprod\rhob(\omega)=k$ for
$\omega\in\Omega^\x$ and
$\v\inprod\rhob(\omega)\le k-1$ for
$\omega\not\in\Omega^\x$.
Thus
indeed $\Omega^\x$ is exposed, and
therefore, by \Thm{exposed}, implicit
submarket closing is always possible.
\end{example}

\begin{example}\emph{Submarket closing on a simplex.}
We show that all events on a simplex are
exposed and thus any random variable
allows implicit submarket closing
by \Thm{exposed}.
Recall that in a market on a simplex,
such as LMSR, we have $\Omega=[K]$ and
$\rho_i(\omega)=\1[i=\omega]$. Let
$\event\subseteq\Omega$ be an arbitrary
event. To see that $\event$ is exposed,
consider $\v\in\reals^K$ with the
components $v_i=\1[i\in\event]$. We have
\[
  \v\inprod\rhob(\omega)=v_\omega=\1[\omega\in\event]
\enspace.
\]
Thus,
$\v\inprod\rhob(\omega)=1$ for
$\omega\in\event$ and
$\v\inprod\rhob(\omega)=0$ for
$\omega\not\in\event$, showing
that $\event$ is exposed.
\end{example}

\subsection{WHEN THE SETS \texorpdfstring{$\hull^x$}{Mx} ARE NOT FACES}
\label{app:ex-roof-nec}

It is worth noting that the condition in Theorem~\ref{thm:one-step-all-goals}
that requires the sets $\hull^x$ to be disjoint faces is merely sufficient and not necessary.
In Figure~\ref{fig:triangle} we give a pictorial example in
two-dimensional price space in which
one of the sets, $\hull^4$,
is not a face of $\hull$, but it is still possible to achieve \CondPrice, \CondVal, and \ZeroVal.

\begin{figure}
  \centering
  \scalebox{1.2}{%
  \begin{tikzpicture}[scale=.8,node distance=1cm]
    \draw (-1,0) node {$\bullet$} -- ++(0.5,0.5) node {$\bullet$};
    \draw (3,0) node {$\bullet$} -- ++(-0.5,0.5) node {$\bullet$};
    \draw (0.5,-2) node {$\bullet$} -- ++(1,0) node {$\bullet$};
    \node at (1.4,-0.5) {$\bullet \; \hull^4$};
    \node at (-1,0.5) {$\hull^1$};
    \node at (3.1,0.5) {$\hull^2$};
    \node at (1,-2.3) {$\hull^3$};
  \end{tikzpicture}%
  }
  \caption{Example showing that the conditions of Theorem~\ref{thm:one-step-all-goals} are not always necessary.}
  \label{fig:triangle}
\end{figure}
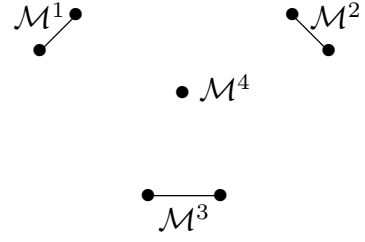


Consider first a market with conditional price spaces $\hull^1$,
$\hull^2$, and $\hull^3$ as shown, but \emph{not} $\hull^4$.  The three sets
$\hull^1,\hull^2,$ and $\hull^3$ are disjoint faces of $\hull$ (the convex
hull of these sets), and hence Theorem~\ref{thm:one-step-all-goals}
applies and \CondPrice, \CondVal and \ZeroVal are satisfied by setting
the new cost function to the conjugate of $\tR = (\conv R^{(\hb^1,\hb^2,\hb^3)})$.
 By construction of $\hbb$ and $\tR$,
the points $(\hmub^x,\tR(\hmub^x))$ for $x\in\set{1,2,3}$ lie on the tangent of $R$
with the slope $\s$, and this same hyperplane is
also a tangent of $\tR$ with the slope $\s$.

Now consider a market with conditional price spaces $\hull^1$,
$\hull^2$, $\hull^3$, and $\hull^4$, as in the figure.  We will argue
that \CondPrice, \CondVal, and \ZeroVal are satisfied for this market
using the conjugate of the same function $\tR$ used
above. First observe that the geometry of $\hull^1,\hull^2,\hull^3,$
and $\hull^4$ implies that regardless of the specific conditional
price vectors $\hmub^x\in\hull^x$ for $x\in\set{1,2,3}$, we always
have that $\hmub^4$ is in the convex hull of $\hmub^1,\hmub^2,$ and $\hmub^3$.
Now by convexity of $\tR$,
the fact that the tangent to $\tR$ with slope $\s$ contains
$(\hmub^x,\tR(\hmub^x))$ for $x\in\set{1,2,3}$
implies that this tangent must also contain the point
$(\hmub^4,\tR(\hmub^4))$. Thus, setting $b^4=R(\hmub^4)-\tR(\hmub^4)$,
we obtain that $\tR$ is consistent with $R^{(\hb^1,\hb^2,\hb^3,b^4)}$
(for the same $\hb^1$, $\hb^2$, and $\hb^3$ as above) which by
\Lem{one-step-condval} guarantees \CondPrice and \CondVal. Since
$(\hmub^4,\tR(\hmub^4))$ is on the tangent, \ZeroVal holds too.
\jenn{In the long version, we may want a lemma about \ZeroVal and the
  tangent since we refer to this informally in lots of places.}

\miro{Killing additional examples.
\\
~\\
An example central to our motivation is the LMSR on the simplex.  One can easily apply Proposition~\ref{prop:extend-R-b} to this case, since any partial information $X$ will simply partition the simplex.  The following example illustrates our construction in this case.
\begin{example}
  \label{ex:one-step-lmsr}
  For the LMSR in Example~\ref{ex:formalism-lmsr}, we might consider
  some $X:\Omega\to[k]$, where $[k]\coloneqq\{1,\ldots,k\}$, which defines a
  $k$-partition of the outcomes into $\Omega^1,\ldots,\Omega^k$.  This
  would correspond to some ``partial'' observation, such as revealing
  the first digit of the winning lottery number.  In this case, we
  can leverage Proposition~\ref{prop:one-step-roof-dual} to write
  \begin{align*}
    \tC(\q) &=\ln\sum_{\omega\in\Omega} e^{s_\omega} + \max_{x\in[k]}
    \left[\ln\sum_{\omega\in\Omega^x} e^{q_\omega} -
      \ln\sum_{\omega\in\Omega^x} e^{s_\omega}\right].
  \end{align*}
  The roof construction essentially breaks the market into several
  LMSR markets, one for each $\Omega^x$.  In fact, one can check that
  $\hmub^x\coloneqq \tpb(X\=x;\s)$ is the vector with $\hmu_\omega =
  e^{s_\omega}/\sum_{\omega'\in\Omega^x}e^{s_{\omega'}}$ for
  $\omega\in\Omega^x$ and 0 otherwise, which is just the probability
  vector $\p(\s)$ conditioned (in the usual probabilistic sense) on
  $X\=x$.  Moreover, the starting price immediately after the
  transition to $\tC$ is the ``superposition'' described by $\tpb(\s)
  = \conv\{\hmub^1,\ldots,\hmub^x\}$.  As in the previous example, a
  trader can collapse the price to $\mub\in\hull^x$ by making a small
  purchase of a security for any outcome $\omega\in\Omega^x$, thus
  implicitly closing the submarket for $X$.
\end{example}
To conclude we remark that it is tempting to think that the $\S\cap\hull^\star=\emptyset$ condition holds whenever the sets $\hull^x$ are extreme in $\hull$, meaning $\hull\setminus\hull^x$ is convex for each $x\in \X$.  This is not the case, as the following example shows: take $\rho(\omega) = \omega$ for $\omega \in \{1,2,3,4\}$, and $X$ such that
\[
  X(\omega)=
\begin{cases}
  0&\text{if $\omega\in\set{1,2}$,}
\\
  1&\text{if $\omega\in\set{3,4}$,}
\end{cases}
\]
i.e., $\hull^0 = [1,2]$ and $\hull^1 = [3,4]$. Both of these sets are extreme in $\hull = [1,4]$, yet not only $\S\cap\hull^\star\ne\emptyset$, but any strictly convex function on $\hull$ will fail the consistency condition of Theorem~\ref{thm:one-step-0-profits} whenever $\p(\s) \in [1,2)\cup(3,4]$.
\\
~\\
{}[[XOR, count example]]\\
Finally, we note that in some situations, achieving \ZeroVal while maintaining \CondVal may not be possible, such as in the following example [[count example on the square]].  In these situations, it may be beneficial to relax \CondVal to \CondPrice, allowing one to still obtain \ZeroVal via the convex roof construction [[illustrate what this looks like for the count example]].  As we saw above, this approach still preserves the worst-case loss bound.}


\section{BOUNDS ON WORST-CASE LOSS}
\label{app:wcl}

In this section, we show that the mechanisms studied in this paper
maintain an important feature of cost-function-based market makers: a
finite bound on the loss of the market maker which is guaranteed to
hold no matter what trades are executed or which outcome $\omega$
occurs.  In particular, we show that the worst-case loss bound of a
market maker using the initial cost function ($C$ for sudden
revelation market makers, $\C(\cdot;t_0)$ for gradual decrease market
makers) is maintained.\footnote{We actually show something slightly
  stronger: for every outcome $\omega$, the worst case loss of the
  market maker conditioned on the true outcome being $\omega$ is
  maintained.}

\miro{Maybe strengthen reasoning below, so that our theorems depend on
  the final $\omega$?}\jenn{Couldn't decide, so I just added a
  footnote for now.}

For a standard cost-function-based market maker with cost function $C$,
the worst-case market maker loss is simply
\begin{align*}
&\WCL(C;\s^\ini)
\\&\quad{}
   \coloneqq
   \!\!
   \sup_{\omega\in\Omega,\r \in \reals^K}
   \bigBracks{\rhob(\omega)\inprod\r - C(\s^\ini+\r) + C(\s^\ini)}
       \label{eq:worst-case-loss}
\end{align*}
where $\s^\ini$ is the initial state of the market.  The term inside
the supremum is the difference between the amount the market maker
must pay traders and the amount collected from traders by the market
maker when the cumulative trade vector is $\r$ and the outcome is
$\omega$.  Our assumption that $\dom R=\hull$, where $R$ is the
conjugate of $C$, guarantees that
$\WCL(C;\s^\ini)$ is always finite~\citep{Abernethy13}.  In
particular, it is easy to see from \Eq{D:mu} of \Thm{D} that
\[
\WCL(C;\s^\ini) = \max_{\omega \in \Omega} D(\rhob(\omega), \s^\ini)
\enspace.
\]
We show that the mechanisms introduced
in Sections~\ref{sec:one-step} and~\ref{sec:gradual} maintain this bound.

\subsection{SUDDEN REVELATION MARKET MAKERS}
\label{app:wcl:one-step}

For sudden revelation market makers (see \Prot{switch}), the
worst-case market maker loss is
\begin{align}
&
  \WCL(C,\AdvanceCost,\AdvanceState;\s^\ini)
\notag
\\
&\quad{}\coloneqq
  \!\!
  \sup_{\omega\in\Omega,\r\in\reals^K,\trb\in\reals^K}\bigl[
     \rhob(\omega)\inprod(\r+\trb)
     -C(\s^\ini+\r)
\notag
\\
&\qquad\qquad\qquad{}
     +C(\s^\ini)-\tC(\tsb+\trb)+\tC(\tsb)
     \bigr]
\label{eq:worst-case-loss-twiddle}
\end{align}
where $\tC = \AdvanceCost(\s^\ini + \r)$ and $\tsb =
\AdvanceState(\s^\ini + \r)$.  Note that $\tC$ and $\tsb$
depend on $\r$ although we do not write this dependence explicitly.
The worst-case loss does not depend on the switch time $t$.

We now bound this worst case loss for our construction in
\Sec{one-step}, with $\tsb$ equal to the state $\s$ at the switch
time and $\tC$ defined to be the conjugate of $\tR = (\conv R^{\hbb})$,
where $\hbb$ depends on $\s$.  We show that the loss of this
market maker is no worse than that of a market maker using the initial
cost function $C$.

\newcommand{\myqi}{\s^{\ini}}
\newcommand{\myqf}{\tsb^{\fin}}
\begin{theorem}
  \label{thm:one-step-wcl}
  If $\AdvanceState(\s)=\s$ and $\AdvanceCost(\s)$ is defined as in \Thm{one-step-0-profits},
  then for any bounded-loss, no-arbitrage cost function $C$ and any
  initial state $\myqi$,
 \[  \thinskips
     \WCL(C\!,\AdvanceCost,\AdvanceState;\myqi)
     \le
     \WCL(C;\myqi)\,
     .\]
\end{theorem}

\begin{proof}
  Let $\myqf$ be the final state of the market and $\s$ be the
  market state at the switch time $t$, as in
  \Prot{switch}.  Then from Proposition~\ref{prop:one-step-roof-dual},
  $\tC(\tsb) = \tC(\s)=C(\s)$ and
  \begin{align*}
  &
    \WCL(C,\AdvanceCost,\AdvanceState;\myqi)
  \\&\quad{}
    =
    \adjustlimits\max_{\omega\in\Omega} \sup_{\myqf
      \in \reals^K} \; \Bigl[ \rhob(\omega) \cdot (\myqf-\myqi)
\\[-0.5\baselineskip]
&\qquad\qquad\qquad\qquad\qquad{}
+ C(\myqi) - \tC(\myqf) \Bigr]
\enspace.
  \end{align*}
By conjugacy we have
\[
\sup_{\myqf \in \reals^K} \Bracks{
  \rhob(\omega) \cdot \myqf -
  \tC(\myqf)} = \tR\bigParens{\rhob(\omega)}
\enspace.
\]
By the definition of $\hbb$,
\[
\tR\bigParens{\rhob(\omega)} = R\bigParens{\rhob(\omega)} - D(\hmub^x\|\s) \leq
  R\bigParens{\rhob(\omega)}
\]
for some
$\hmub^x\in\p(\Omega^x;\s)$ where
$x \in \X$ is such that $\omega \in \Omega^x$.  Putting this together, we obtain
the bound
\begin{align}
\notag
&
     \WCL(C,\AdvanceCost,\AdvanceState;\myqi)
\\
\notag
&\quad{}
    \leq
    \max_{\omega\in\Omega} \Bracks{R\bigParens{\rhob(\omega)} + C(\myqi)
    - \rhob(\omega)\cdot \myqi}
\\
\tag*{\qed}
&\quad{}
    = \WCL(C;\myqi)
\enspace.
\end{align}
\renewcommand{\qed}{}
\end{proof}

\subsection{GRADUAL DECREASE LCMMS}
\label{app:wcl:gradual}

For gradual decrease market makers (see \Prot{time-sensitive}), the worst-case
market maker loss can be written as
\begin{align}
&
  \WCL(\C,\AdvanceState;\s^0,t^0)
\notag
\\
&\quad{}\coloneqq
  \!\!
  \sup_{\substack{
       \omega\in\Omega,N\ge0,\set{\r^i}_{i=1}^N,\set{t^i}_{i=1}^N\\
       \text{with }t^0\le t^1\le\dotsb\le t^N}}
     \left[
     \sum_{i=1}^N\bigl[
       \rhob(\omega)\inprod\r^i
     \right.
\notag
\\
&\qquad\qquad\qquad{}
       -\C(\tsb^{i-1}+\r^i;t^i)+\C(\tsb^{i-1};t^i)
     \bigr]
     \Bigr]
  \label{eq:wcl:gradual}
\end{align}
where $\tsb^{i-1} = \AdvanceState(\s^{i-1}; t^{i-1}, t^i)$.

We next show that the worst-case loss of the gradual decrease LCMM
developed in \Sec{gradual} is no worse than that of a market maker
using the initial cost function $\C(\cdot;t^0)$.

\begin{theorem}
  For the gradual decrease LCMM with corresponding function
  \AdvanceState\ and cost $\C$ and any differentiable non-increasing
  information-utility schedules $\beta_g$, for any initial state $\s_0$
  and time $t_0$,
\[
 \WCL(\C,\AdvanceState;\s^0,t^0)
 \le
 \WCL(C^0;\s^0)
\]
where $C^0\coloneqq\C(\cdot;t^0)$.
\end{theorem}
\begin{proof}
In the context of \Prot{time-sensitive}, let $C^i$
denote $\C(\cdot;t^i)$, and $R^i$ and $D^i$ denote
the corresponding conjugate and divergence. First, note that
by \Thm{time-sensitive}, for any $i$ and any $\mub\in\hull$, for
suitable $\deltab^\star$ and $\etab^\star$,
\begin{align}
&
  D^{i+1}(\mub\|\tsb^i)
\notag
\\
&\quad{}
  =
  \sum_{g\in\G}
  \frac{\beta_g(t^{i+1})}{\beta_g(t^i)}
  D^i_g(\mub_g\|\s^i_g+\deltab_g^\star)
  +(\A^\top\mub-\b)\inprod\etab^\star
\notag
\\
&\quad{}
  \le
  \sum_{g\in\G}
  D^i_g(\mub_g\|\s^i_g+\deltab_g^\star)
  +(\A^\top\mub-\b)\inprod\etab^\star
\notag
\\
&\quad{}
  =
  D^i(\mub\|\s^i)
\enspace.
\label{eq:WCL:3}
\end{align}
The last equality follows from \Thm{lcmm}\ref{lcmm:D}.

We can bound the expression inside the supremum in \Eq{wcl:gradual} as
{\thinskips%
\begin{align}
&
\sum_{i=1}^N
     \BigBracks{
       \rhob(\omega)\inprod\r^i
       -C^i(\tsb^{i-1}\!+\r^i)+C^i(\tsb^{i-1})
     }
\notag
\\
&\quad{}=
\sum_{i=1}^N
     \Bigl[
       R^i\bigParens{\rhob(\omega)} + C^i(\tsb^{i-1})
       -\rhob(\omega)\inprod\tsb^{i-1}
\notag
\\
&\qquad\qquad{}
       -R^i\bigParens{\rhob(\omega)} - C^i(\tsb^{i-1}\!+\r^i)
       +\rhob(\omega)\inprod(\tsb^{i-1}\!+\r^i)
     \Bigr]
\notag
\\
&\quad{}=
\sum_{i=1}^N
     \BigBracks{
       D^i(\rhob(\omega)\|\tsb^{i-1})
       -D^i(\rhob(\omega)\|\tsb^{i-1}\!+\r^i)
     }
\notag
\\
&\quad{}=
\sum_{i=1}^N
     \BigBracks{
       D^i(\rhob(\omega)\|\tsb^{i-1})
       -D^i(\rhob(\omega)\|\s^i)
     }
\notag
\\
&\quad{}=
D^1(\rhob(\omega)\|\tsb^0)
+ \sum_{i=1}^{N-1}
     \BigBracks{
       D^{i+1}(\rhob(\omega)\|\tsb^i)
       -D^i(\rhob(\omega)\|\s^i)
     }
\notag
\\
&\qquad\qquad{}
-D^N(\rhob(\omega)\|\s^N)
\notag
\\
&\quad{}\le
D^0(\rhob(\omega)\|\s^0)
\enspace
\notag
\end{align}}%
where the last inequality follows by applications of \Eq{WCL:3} to the
first two terms
and the positivity of $D^N(\cdot\|\cdot)$. Taking the supremum,
we obtain
\begin{align}
&
  \WCL(\C,\AdvanceState;\s^0,t^0)
\notag
\\
&\quad{}
  \le
  \max_{\omega\in\Omega} D^0(\rhob(\omega)\|\s^0)
  =
  \WCL(C^0;\s^0)
\enspace.
\tag*{\qed}
\end{align}
\renewcommand{\qed}{}
\end{proof}



\end{document}